\documentclass[runningheads]{llncs}
\usepackage[T1]{fontenc}
\usepackage{graphicx}
\usepackage{soul}
\usepackage{url}
\usepackage{hyperref}
\usepackage[utf8]{inputenc}
\usepackage{caption}
\usepackage{graphicx}
\usepackage{amsmath}
\usepackage{booktabs}
\usepackage{bm}
\usepackage{soul}

\usepackage{amsfonts}
\usepackage{comment}
\usepackage{stmaryrd}
\usepackage{xspace}
\usepackage{csquotes}
\usepackage{booktabs}
\usepackage{multirow}

\usepackage{backnaur}
\usepackage{oplotsymbl}
\usepackage{algorithm}
\usepackage{algorithmicx}
\usepackage{algpseudocode}
\usepackage{tcolorbox}
\usepackage{titletoc}
\usepackage{thmtools}
\usepackage{thm-restate}

\usepackage{subcaption}
\usepackage[font=small,labelfont=bf,tableposition=top]{caption}

\usepackage{color}

\usepackage{appendix}

\algnotext{EndFor}
\algnotext{EndIf}
\algnotext{EndWhile}

\usepackage{tikz}
\usetikzlibrary{arrows,arrows.meta,calc}
\usetikzlibrary{shadows}
\usepackage{enumerate}
\usepackage{xspace}
\usepackage{tikz-er2}
\usepackage{tikz-orm}
\usetikzlibrary{shapes.multipart}
\usetikzlibrary{positioning}
\usetikzlibrary{shapes.multipart}
\usepackage{xcolor-solarized}

\usepackage{pgfplots}
\usepackage{pgfplotstable}

\tikzset{
    POINTS/.style={solid, mark=square*, red},
    V1COALESCED/.style={solid, mark=square, blue},
    V2COALESCED/.style={solid, mark=o, green},
    V1NONCOALESCED/.style={solid, mark=square, blue},
    V2NONCOALESCED/.style={solid, mark=o, green},
}

\pgfplotsset{
    every axis/.append style={
        legend columns=-1,
        legend style={
            at={(0.5,1)},
            anchor=south,
            font=\tiny,
            draw=none,
            fill=none,
            row sep=0.01cm
        },
        height = 3.5cm,
        width = 4.3cm,
        xlabel near ticks,
        ylabel near ticks,
        legend cell align = left,
        label style = {font=\scriptsize},
    },
    every tick label/.append style={
        font=\tiny
    },
}

\usepackage{adjustbox}
\usepackage{dashbox}

\newif\ifdraft
\drafttrue 

\ifdraft
  \usepackage[textsize=scriptsize,colorinlistoftodos]{todonotes}
\else
  \usepackage[disable]{todonotes}
\fi

\definecolor{basicgreen}{rgb}{0.0, 0.6, 0.0}
\definecolor{strongblue}{rgb}{0.1, 0.1, 1.0}
\definecolor{strongred}{rgb}{0.0, 0.0, 1.0}

\definecolor{mygreen}{rgb}{0,0.5,0.0}
\definecolor{mygray}{rgb}{0.5,0.5,0.5}
\definecolor{mypurple}{rgb}{0.38,0,0.32}
\definecolor{myblue}{rgb}{0.1,0,0.32}

\definecolor{darkgrey}{RGB}{70,70,70}
\definecolor{lightgrey}{RGB}{200,200,200}

\newcommand{\edgelabel}{\mathit{label}}

\newcommand{\fg}{\mathcal{F}_G}
\newcommand{\fgp}{\mathcal{F}_{G'}}
\newcommand{\nodesg}{\N_G}

\newcommand{\infi}{\mathsf{inf}}
\newcommand{\supr}{\mathsf{sup}}

\newcommand{\bound}[2]{\mathsf{boundaries}(#1,#2)}

\newcommand{\ratk}{\frac{1}{k} \integers}
\newcommand{\rattk}{\frac{1}{2k} \integers}

\newcommand{\ei}{\emph{(i)~}}
\newcommand{\eii}{\emph{(ii)~}}
\newcommand{\eiii}{\emph{(iii)~}}

\newcommand{\valz}[1][G]{\mathsf{val}_{#1}}
\newcommand{\val}[3]{\valz(#1, #2, #3)}
\newcommand{\valpz}{\valz[G']}
\newcommand{\valp}[3]{\valpz(#1, #2, #3)}

\newcommand{\dom}[1]{\mathsf{dom}(#1)}
\newcommand{\range}[1]{\mathsf{range}(#1)}

\newcommand{\PSPACE}{\textsc{PSpace}\xspace}

\newcommand{\NP}{\textsc{NP}\xspace}

\newcommand{\sizea}[2]{\#\mathsf{answers}^{#1}(#2)}

\newcommand{\td}{\mathcal{T}}
\newcommand{\tdg}{\td_{G}}

\newcommand{\tuples}{\mathcal{U}}
\newcommand{\tuplest}{\tuples^t}
\newcommand{\tuplesd}{\tuples^d}
\newcommand{\tuplestd}{\tuples^{td}}
\newcommand{\tuplestdbe}{\tuples^c}
\newcommand{\tuplesx}{\tuples^{x}}

\newcommand{\pbma}{\textsc{Answer}\xspace}
\newcommand{\pbmz}{\textsc{CompactAnswer}}
\newcommand{\pbm}[1]{\pbmz$^{#1}$\xspace}
\newcommand{\pbms}{\pbm{}}
\newcommand{\pbmt}{\pbm{t}}
\newcommand{\pbmd}{\pbm{d}}
\newcommand{\pbmx}{\pbm{x}}
\newcommand{\pbmtd}{\pbm{td}}
\newcommand{\pbmtdbe}{\pbm{c}}

\newcommand{\nn}{\mathbb{N}}
\newcommand{\integers}{\mathbb{Z}}

\DeclareFontFamily{OMX}{MnSymbolE}{}
\DeclareFontShape{OMX}{MnSymbolE}{m}{n}{
    <-6>  MnSymbolE5
   <6-7>  MnSymbolE6
   <7-8>  MnSymbolE7
   <8-9>  MnSymbolE8
   <9-10> MnSymbolE9
  <10-12> MnSymbolE10
  <12->   MnSymbolE12}{}
\DeclareSymbolFont{mnlargesymbols}{OMX}{MnSymbolE}{m}{n}
\SetSymbolFont{mnlargesymbols}{bold}{OMX}{MnSymbolE}{b}{n}
\DeclareMathDelimiter{\llangle}{\mathopen}{mnlargesymbols}{'164}{mnlargesymbols}{'164}
\DeclareMathDelimiter{\rrangle}{\mathclose}{mnlargesymbols}{'171}{mnlargesymbols}{'171}

\newcommand{\nt}{\noindent}

\newcommand{\intervals}[1]{\mathsf{intv}(#1)}

\newcommand{\ld}[1]{_{#1}\!\lfloor}
\newcommand{\rd}[1]{\rfloor\!_{#1}}

\newcommand{\evalci}[2][G]{\llparenthesis #2 \rrparenthesis_{#1}}
\newcommand{\evalcit}[2][G]{\evalci[#1]{#2}^t}
\newcommand{\evalcid}[2][G]{\evalci[#1]{#2}^d}
\newcommand{\evalcitd}[2][G]{\evalci[#1]{#2}^{td}}
\newcommand{\evalcitdbe}[2][G]{\evalci[#1]{#2}^c}

\newcommand{\eval}[2][G]{\llbracket #2 \rrbracket_{#1}}
\newcommand{\evalp}[2][G']{\llbracket #2 \rrbracket_{#1}}

\newcommand{\edge}{\mathsf{edge}}
\newcommand{\node}{\mathsf{node}}
\newcommand{\pred}{\mathit{pred}}
\newcommand{\query}{\mathsf{trpq}}

\newcommand{\tnz}{\textbf{T}}
\newcommand{\tn}[1]{\tnz_{#1}}
\newcommand{\tndelta}{\tnz_\delta}

\newcommand{\te}[1]{\text{\ #1\ }}

\newcommand{\compl}[2]{\mathsf{compl}\left(#1,#2\right)}

\newcommand{\join}{\mathbin{\bowtie}}
\newcommand{\tjoin}{\mathbin{\overline{\bowtie}}}

\newcommand{\tup}[1]{\langle #1\rangle}            
\newcommand{\Tup}[1]{\big\langle #1\big\rangle}

\newcommand{\rat}{\mathbb{Q}}

\newcommand{\para}[1]{\vspace{+0.1cm} \noindent\textbf{#1.}}

\newcommand{\mc}[1]{\mathcal{#1}}

\newcommand{\E}{\mc{E}}

\newcommand{\N}{\mc{N}}

\renewcommand{\u}{\mathbf{u}}
\newcommand{\vt}{\mathbf{v}}

\newcommand{\keyword}[1]{\textcolor{basicgreen}{ \mathtt{#1}} \xspace}

\newcommand{\joinsymbol}{\textcolor{red}{\bm{/}} \xspace}

\newcommand{\scaleExamples}{0.6}
\newcommand{\scaleExamplesSecondFig}{0.75}
\newcommand{\roundingEdgesExamples}{5pt}

\usepackage{lstautogobble}
\usepackage{listings}

\makeatletter
\newcommand{\lstuppercase}
{\uppercase\expandafter{\expandafter\lst@token\expandafter{\the\lst@token}}}

\newcommand{\lstlowercase}
{\lowercase\expandafter{\expandafter\lst@token\expandafter{\the\lst@token}}}
\makeatother

\lstdefinestyle{SQLstyle}{%
  language=SQL,%
  basicstyle=\ttfamily,%
  keywordstyle=\color{blue}\ttfamily\lstuppercase\bfseries,%
  morekeywords={and, order, by, access, mod, nls_date_format, nvl,
    replace, sysdate, to_char, to_number, trunc, references, with,
    returns, begin, query, return, language, plpgsql, declare,
    function, loop, if, for, rowtype, strict, new, setof, after,
    each, row, procedure, raise, before, over, partition, unbounded,
    preceding, rank, lag, century, to, materialized, type, sequence,
    public, database, rule, tablespace, context, dimension, user,
    synonym, privileges, privilege, operator, flashback, export,
    import, object, link, dictionary, inherits, less, than,
    determines, old, partitions, hash, store, reference, reformat,
    fold, unfold, cp, dur, valid, seq, nseq, vt, now, duration,
    snapshot, normalize, align, absorb, daterange, analyze, explain, timing, off, lead, window},%
  deletekeywords={END},%
  identifierstyle=\ttfamily,%
  stringstyle=\ttfamily,%
  showstringspaces=false,%
  mathescape=true,%
  upquote=true,%
}

\lstnewenvironment{SQL}[1][] {
  \lstset{
    style=SQLstyle,
    commentstyle=\color{gray},
    autogobble=true,
    xleftmargin=10pt,
    xrightmargin=10pt,
    framesep=2pt,
    #1
  }}{}

\begin{document}
%

\title{Compact Answers to Temporal Path Queries}

\author{
Muhammad Adnan\inst{1} \and
Diego Calvanese\inst{1} \and
Julien Corman\inst{1} \and
Anton Dignös\inst{1} \orcidID{0000-0002-7621-967X} \and
Werner Nutt\inst{1}\orcidID{0000-0002-9347-1885} \and 
Ognjen Savković\inst{1}
}

\authorrunning{M. Adnan et al.}

\institute{Free University of Bozen-Bolzano, Italy}

\maketitle              
\begin{abstract}

We study path-based graph queries that, in addition to navigation through
edges, also perform navigation through time.  
This allows asking questions about the dynamics of
networks, like traffic movement, cause-effect relationships, or the spread
of a disease.
In this setting, a graph consists of triples annotated with validity
intervals, and a query produces pairs of nodes where each pair is
associated with a binary relation over time.  For instance, such a pair
could be two airports,
and the temporal relation could map potential departure times
to possible arrival times.
An open question is how to represent such a relation in a compact form
and maintain this property during query evaluation.
We investigate four compact representations of answers to a such queries, 
which are based on alternative ways to encode sets of intervals.
We discuss their respective advantages and drawbacks, in terms of conciseness,
uniqueness, and computational cost.
Notably, the most refined encoding guarantees that query answers over dense time
can be finitely represented.


\end{abstract}
%


\input{toc}

\section{Introduction}
\label{sec:intro}

Temporal databases~\cite{DBLP:conf/ebiss/BohlenDGJ17} have traditionally focused on computing relations
where each tuple is annotated with a \emph{single} time point (or single interval) for validity.
In comparison, little attention has been paid to relations over time points,
which indicate how different events are related temporally.
This is the focus of this paper, more precisely binary relations over time points.
%
 
%
One application is linking potential departure times of a road trip to their corresponding arrival times, considering uncertainties like traffic.
Such a relation matches each departure time with a range of possible arrival times (or equivalently each potential arrival time back to possible departures).
%
%
%
Another example is modelling cause-effect scenarios with uncertain delays.
For instance, one may compute the relation that associates (in time) the potential malfunction a component in an airplane to the subsequent malfunction of another component. 
%
A third application of binary temporal relations, which we will use as a running example,
is modeling the spread of phenomena such as messages or diseases.
For instance, in epidemiology, the latency and infectiousness periods of a virus induce a temporal relation that maps potential infection times to possible subsequent transmissions, which allows modeling disease propagation within a population.

These examples suggest at least three elementary operations that one may want to perform over such relations: 
\ei{} \emph{filtering} a temporal relation (based on some knowledge about events that took place),
\eii \emph{combining} two temporal relations,
and \eiii \emph{composing} them (e.g. to model transitive cause-effect or spread).
Or in database terms, selection, union and join respectively.
In particular, a fundamental question is how one can \emph{represent} a temporal binary relation to start with,
in a compact way (or simply in a finite way over dense time),
and whether compactness (resp. finiteness) is preserved by selection, union and/or join.
This is the main problem addressed in this paper.

\begin{figure}[htbp]
    \centering






\begin{tikzpicture}
[
      NodeRed/.style={
        rounded corners,
        fill=red!80!black!15!white,
        draw=green!20!black,
        semithick
      },
      NodeSolar/.style={
        rounded corners,
        fill=solarized-base3!50!white,
        draw=green!20!black,
        semithick
      },
      NodeGreen/.style={
        rounded corners,
        fill=green!80!black!15!white,
        draw=green!20!black,
        semithick
      },
      ISA/.style={
        thick,-{Triangle[open]},
      },
      Edge/.style={
        thick,>=latex,->
      },
      EdgeProperties/.style={
        sloped,
        below,
        rounded corners,
        fill=solarized-base3!50!white,
      }
      ]

\node[NodeGreen] (alice) {\emph{Alice}};
\node[NodeGreen] (iswc) [right=3.5cm of alice] {\emph{ISWC}};
\node[NodeGreen] (icdt) [below=of iswc] {\emph{ICDT}};

\node[NodeGreen] (bob)  [right=3.5cm of iswc] {\emph{Bob}};
\node[NodeGreen] (positive) [below=of bob] {\emph{positive}};

\draw[Edge, bend right=15] 
  (alice) to  node[above, fill=yellow!20, rounded corners, inner sep=2pt] 
  {$\mathtt{attends}$  \textbf{[100, 102]}} (icdt);

\draw[Edge] 
  (alice) -- node[above, fill=yellow!20, rounded corners, inner sep=2pt] 
  {$\mathtt{attends}$  \textbf{[104, 106]}} (iswc);

\draw[Edge] 
  (bob) -- node[above, fill=yellow!20, rounded corners, inner sep=2pt] 
  {$\mathtt{attends}$ \textbf{[102, 107]}} (iswc);


\draw[Edge] 
  (bob) -- node[left, fill=yellow!20, rounded corners, inner sep=2pt] 
  {$\mathtt{tests}$ \textbf{[112, 112]}} (positive);

\end{tikzpicture}

    \caption{A temporal graph}
  \label{fig:graph}
\end{figure}
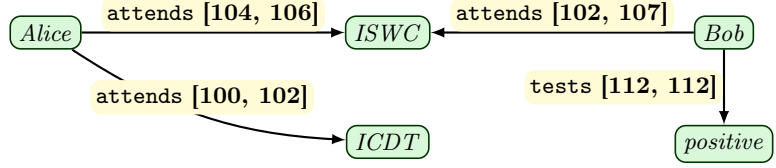

We focus on the query language recently introduced in~\cite{arenas2022temporal}, which computes and performs operations on binary temporal relations, along with a graph data model used to filter such relations.
For presentation purposes, we simplify this language and data model to a fragment that is essentially Regular Path Queries (RPQs) extended with a temporal navigation operator,
and evaluated over a time-labeled graphs.
\footnote{RPQs are a central building block of navigational graph query languages such as Cypher~\cite{francis2018cypher} and SPARQL~\cite{harris2013sparql}.}
RPQs are used to navigate through the graph,
and the temporal operator enables navigation in time by a certain \emph{range}.
%

An RPQ $q$ is a regular expression, and a pair
$\tup{n_1, n_2}$ of nodes is an answer to $q$ over a graph $G$ if there exists a
path from $n_1$ to $n_2$ in $G$ whose concatenated labels match this regular expression.
The language of~\cite{arenas2022temporal}, 
called \emph{Temporal RPQs} (TRPQs),
extends RPQs with a temporal operator that 
allows navigation from one node at a certain time to the same node in past or future moments.
Accordingly,
an answer is a pair $\tup{\tup{n_1, t_1}, \tup{n_2, t_2}}$,
where $n_1$ and $n_2$ are each associated with a time point,
which mark the starting and arrival time of the temporal navigation respectively.

More precisely,
this temporal operator, which we denote as $\tndelta$ here,
allows navigation in time within the range
specified as the interval $\delta$.
For instance,
the latency and infectiousness periods of a virus may imply that a person can only become infectious at least three days after the infection time $t$,
and remain infectious at most five days after $t$.
If $\td$ is our temporal domain,
this induces a binary temporal relation 
\[R = \{(t, t + d) \mid t \in \td \te{and} d \in [3, 5]\}.\]

Accordingly, if $\nodesg$ is the set of nodes in the queried graph $G$,
then the query $\tn{[3, 5]}$ outputs all pairs $\tup{\tup{n, t_1}, \tup{n, t_2}}$ such that 
$n \in \nodesg$ and $(t_1, t_2) \in R$.
This output can be joined and filtered based on events registered in $G$,
effectively restricting the initial temporal relation $R$,
and associating each returned time point with a meaningful node.
For instance, the graph of Figure~\ref{fig:graph} represents data about attendance at a conference, where each fact is a labelled edge,
annotated with an interval for validity.
We can extend our query as follows

\[q_1 = \keyword{\tn}_{[3, 5]} \joinsymbol \mathtt{attends} \joinsymbol \mathtt{attends}^-.\]

This query outputs all pairs $\tup{\tup{n_1, t_1}, \tup{n_2, t_2}}$ such that $n_1$ and $n_2$ attended a same event,
and $n_1$ may have transmitted the virus to $n_2$ (during this event) at time $t_2$ if $n_1$ got infected at time $t_1$.
These include the two pairs 
$\langle \tup{\mathit{Alice}, 100}$, $ \tup{\mathit{Bob}, 104}\rangle$,
and
$\tup{\tup{\mathit{Alice}, 100}, \tup{\mathit{Bob}, 105}}$,
meaning that if Alice contracted the virus on day 100,
then she may have transmitted it to Bob during day 104 or day 105 (among other possibilities).
In this query, the "$\joinsymbol$" operator is a join, and the subqueries $\mathtt{attends}$ and $\mathtt{attends}^-$ act as filters on our initial temporal relation $R$,
restricting it based on the time intervals contained in the graph.
We can further restrict this relation by requiring that person $n_2$ tested positive at most a week after the possible transmission, with

\[q_2 = q_1 \joinsymbol ?\Big(\keyword{\tn}_{[0, 7]} \joinsymbol \mathtt{tests} \joinsymbol \mathit{(= positive)}\Big).\]

In this query, the operator $?q$ acts as an existential quantifier: it only requires the existence of a match for the subquery $q$,
so that the output of $q_2$ is a subset of the output of $q_1$.

As we have seen above, the initial temporal relation $R$ can be easily represented in a compact way,
using the temporal domain $\td$ for the set of initial infection points,
and a single interval of temporal distances $[3,5]$ that indicates by how much each initial time point can be shifted.
This holds regardless of whether $\td$ is discrete or continuous.
However, this does not hold anymore for the outputs of our queries.
For instance, 
the temporal relation induced by the answers to $q_2$ (which all have \textit{Alice} and \textit{Bob} as first and second node respectively)
is uneven:
over discrete time (and with a granularity of one day),
100 is associated with a single day 105,
whereas 101 is associated with both 105 and 106.
%
Moreover, without a compact representation of query answers,
altering the time granularity (e.g. from days to hours) may significantly increase the number of output tuples (exponentially in the size of the input, assuming that interval boundaries are encoded in binary).
Besides, even in the case where the output of a query can be represented compactly,
computing all answers before summarizing them may be prohibitive,
because the number of intermediate results may impact the performance of (worst-case quadratic) join operations.
Hence the need to not only represent answers in a compact way, but also maintain compactness during query evaluation.
This is even a necessity over dense time, where the lack of a finite representation may forbid query evaluation.
To our knowledge, 
these are still open questions (in particular, left open in~\cite{arenas2022temporal}),
which effectively rule out arbitrary selection, union and joins over binary temporal relations.

\para{Contributions and Organization}
In this work, we propose four alternative compact representations of answers to TRPQs.
In Section~\ref{sec:contributions}, we provide an informal, graphical overview of each of them.
In Section~\ref{sec:preliminaries}, we define the syntax and semantics of the query language that we study.
In Section~\ref{sec:compact}, we formally define and study our four representations: first, in Section~\ref{sec:compact_ep}, 
we present the two simpler ones,
which aggregate tuples along a single time dimension, either starting points or distances;
then, in Section~\ref{sec:foldBoth}, we present the two more complex representations,
which aggregate tuples along both dimensions.
We analyze the respective advantages and drawbacks of each representation in terms of finiteness (over dense time), compactness (when finite), uniqueness, and the computational cost of query answering and minimizing a set of tuples.
Notably, the fourth (more complex) representation guarantees that compactness and finiteness are maintained throughout query evaluation.
In Section~\ref{sec:related-work}, we discuss related work, and conclude in Section~\ref{sec:conclusions}.



\section{Overview}
\label{sec:contributions}
This section provides a high-level overview of the four compact representations of answers to TRPQs defined and studied in this article.
As a running example, 
the following query $q_3$ outputs all pairs $\tup{\tup{n_1, t_1}, \tup{n_2, t_2}}$ such that,
if Alice got infected at time $t_1$ while attending event $n_1$,
then she may have transmitted the virus at time $t_2$ while attending event $n_2$:
\[q_3 = \mathtt{attends^-} \joinsymbol (=\mathit{Alice}) \joinsymbol \keyword{\tn}_{[3,5]} \joinsymbol \mathtt{attends}.\]

Over discrete time, with days as temporal granularity,
evaluating $q_3$ over the graph $G$ of Figure~\ref{fig:graph} outputs a set $\eval{q_3}$ of $7$ pairs, 
each of the form $\tup{\tup{\mathit{ICDT}, t_1}, \tup{\mathit{ISWC}, t_2}}$.
Since the two nodes (\textit{ICDT} and \textit{ISWC}) are identical for all answers,
we can focus on the binary temporal relation induced by these, 
i.e. the relation $R = \{(t_1, t_2) \mid \tup{\tup{\mathit{ICDT}, t_1}, \tup{\mathit{ISWC}, t_2}} \in \eval{q_3}\}$.
Each pair $(t_1, t_2) \in R$ can equivalently be represented as the pair $(t_1, t_2 - t_1)$,
where the second component is the temporal distance between $t_1$ and $t_2$.
Accordingly, Figure~\ref{fig:all_answers} displays these 7 pairs over the Euclidean plane of (discrete) time per distance.

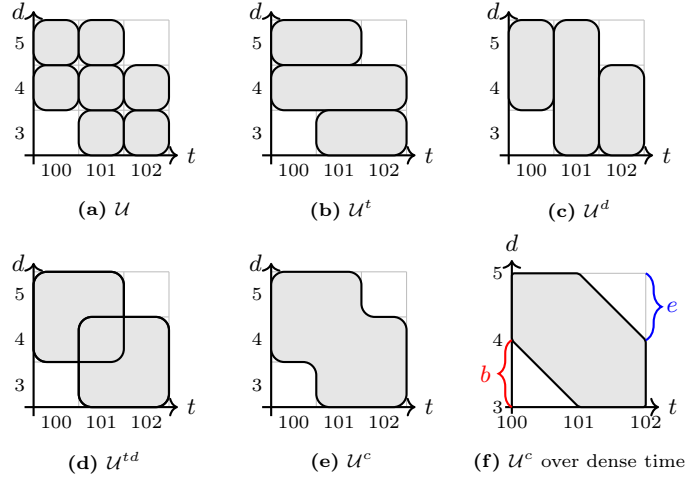
\begin{figure}[htbp]
  \centering
   \begin{subfigure}[t]{0.25\textwidth}
    \centering
      \begin{tikzpicture}[scale=\scaleExamples, rounded corners=\roundingEdgesExamples]

  \foreach \x in {0,1,2,3} {
    \draw[gray!50] (\x, 0) -- (\x, 3);
  }

  \foreach \y in {0,1,2,3} {
    \draw[gray!50] (0, \y) -- (3, \y);
  }

  \foreach \i/\label in {0.5/100, 1.5/101, 2.5/102} {
    \node[below] at (\i, 0) {\scriptsize \label};
  }

  \foreach \j/\label in {0.5/3, 1.5/4, 2.5/5} {
    \node[left] at (0, \j) {\scriptsize  \label};
  }

  \draw[->, thick] (-0.2, 0) -- (3.2, 0) node[right] {$t$};
  \draw[->, thick] (0, -0.2) -- (0, 3.2) node[left] {$d$};

  \foreach \x/\y in {0/1, 0/2, 1/0, 1/1, 1/2, 2/0, 2/1} {
    \filldraw[fill=gray!20, thick, draw=black] (\x,\y) rectangle ++(1,1);
  }


\end{tikzpicture}
    \caption{ $\tuples$}
    \label{fig:all_answers}
  \end{subfigure}
  \begin{subfigure}[t]{0.25\textwidth}
    \centering
      \begin{tikzpicture}[scale=\scaleExamples]

  \foreach \x in {0,1,2,3} {
    \draw[gray!50] (\x, 0) -- (\x, 3);
  }
  \foreach \y in {0,1,2,3} {
    \draw[gray!50] (0, \y) -- (3, \y);
  }

  \foreach \i/\label in {0.5/100, 1.5/101, 2.5/102} {
    \node[below] at (\i, 0) {\scriptsize \label};
  }

  \foreach \j/\label in {0.5/3, 1.5/4, 2.5/5} {
    \node[left] at (0, \j) {\scriptsize \label};
  }

  \draw[->, thick] (-0.2, 0) -- (3.2, 0) node[right] {$t$};
  \draw[->, thick] (0, -0.2) -- (0, 3.2) node[left] {$d$};

  \path[fill=gray!20, thick, draw=black, rounded corners=\roundingEdgesExamples] (0,2) rectangle (2,3); 
  \path[fill=gray!20, thick, draw=black, rounded corners=\roundingEdgesExamples] (0,1) rectangle (3,2); 
  \path[fill=gray!20, thick, draw=black, rounded corners=\roundingEdgesExamples] (1,0) rectangle (3,1); 

\end{tikzpicture}
    \caption{ $\tuplest$}
    \label{fig:cover_t}
  \end{subfigure}
  \begin{subfigure}[t]{0.25\textwidth}
    \centering
      \begin{tikzpicture}[scale=\scaleExamples]

  \foreach \x in {0,1,2,3} {
    \draw[gray!50] (\x, 0) -- (\x, 3);
  }
  \foreach \y in {0,1,2,3} {
    \draw[gray!50] (0, \y) -- (3, \y);
  }

  \foreach \i/\label in {0.5/100, 1.5/101, 2.5/102} {
    \node[below] at (\i, 0) {\scriptsize \label};
  }

  \foreach \j/\label in {0.5/3, 1.5/4, 2.5/5} {
    \node[left] at (0, \j) {\scriptsize \label};
  }

  \draw[->, thick] (-0.2, 0) -- (3.2, 0) node[right] {$t$};
  \draw[->, thick] (0, -0.2) -- (0, 3.2) node[left] {$d$};

  \path[fill=gray!20, thick, draw=black, rounded corners=\roundingEdgesExamples] (0,1) rectangle (1,3); 
  \path[fill=gray!20, thick, draw=black, rounded corners=\roundingEdgesExamples] (1,0) rectangle (2,3); 
  \path[fill=gray!20, thick, draw=black, rounded corners=\roundingEdgesExamples] (2,0) rectangle (3,2); 

\end{tikzpicture}
    \caption{$\tuplesd$}
    \label{fig:cover_d}
  \end{subfigure}
\\
  \begin{subfigure}[t]{0.25\textwidth}
    \centering
     \begin{tikzpicture}[scale=\scaleExamples]

  \foreach \x in {0,1,2,3} {
    \draw[gray!50] (\x, 0) -- (\x, 3);
  }
  \foreach \y in {0,1,2,3} {
    \draw[gray!50] (0, \y) -- (3, \y);
  }

  \foreach \i/\label in {0.5/100, 1.5/101, 2.5/102} {
    \node[below] at (\i, 0) {\scriptsize \label};
  }

  \foreach \j/\label in {0.5/3, 1.5/4, 2.5/5} {
    \node[left] at (0, \j) {\scriptsize \label};
  }

  \draw[->, thick] (-0.2, 0) -- (3.2, 0) node[right] {$t$};
  \draw[->, thick] (0, -0.2) -- (0, 3.2) node[left] {$d$};

  \path[fill=gray!20, thick, draw=black, rounded corners=\roundingEdgesExamples] (0,1) rectangle (2,3); 
  \path[fill=gray!20, thick, draw=black, rounded corners=\roundingEdgesExamples] (1,0) rectangle (3,2); 

  \draw[draw=black, thick, rounded corners=\roundingEdgesExamples] (0,1) rectangle (2,3); 
  \draw[draw=black, thick, rounded corners=\roundingEdgesExamples] (1,0) rectangle (3,2); 

\end{tikzpicture}
    \caption{$\tuplestd$}
    \label{fig:cover_td}
  \end{subfigure}
  \begin{subfigure}[t]{0.25\textwidth}
    \centering
      \begin{tikzpicture}[scale=\scaleExamples]

  \foreach \x in {0,1,2,3} {
    \draw[gray!50] (\x, 0) -- (\x, 3);
  }
  \foreach \y in {0,1,2,3} {
    \draw[gray!50] (0, \y) -- (3, \y);
  }

  \foreach \i/\label in {0.5/100, 1.5/101, 2.5/102} {
    \node[below] at (\i, 0) {\scriptsize \label};
  }

  \foreach \j/\label in {0.5/3, 1.5/4, 2.5/5} {
    \node[left] at (0, \j) {\scriptsize \label};
  }

  \draw[->, thick] (-0.2, 0) -- (3.2, 0) node[right] {$t$};
  \draw[->, thick] (0, -0.2) -- (0, 3.2) node[left] {$d$};

  \draw[fill=gray!20, draw=black, thick, rounded corners=\roundingEdgesExamples] 
    (0,1) -- (0,3) -- (2,3) -- (2,2) -- (3,2) -- (3,0) -- (1,0)-- (1,1)--cycle; 

\end{tikzpicture}
    \caption{$\tuplestdbe$}
    \label{fig:cover_tdbe}
  \end{subfigure}
  \begin{subfigure}[t]{0.25\textwidth}
    \centering
      \begin{tikzpicture}[scale=0.89]

  \foreach \x in {0,1,2} {
    \draw[gray!50] (\x, 0) -- (\x, 2);
  }
  \foreach \y in {0,1,2} {
    \draw[gray!50] (0, \y) -- (2, \y);
  }

  \foreach \i/\label in {0/100, 1/101, 2/102} {
    \node[below] at (\i, 0) {\scriptsize \label};
  }

  \foreach \j/\label in {0/3, 1/4, 2/5} {
    \node[left] at (0, \j) {\scriptsize \label};
  }

  \draw[->, thick] (-0.1, 0) -- (2.2, 0) node[right] {$t$};
  \draw[->, thick] (0, -0.1) -- (0, 2.2) node[above] {$d$};

  \draw[fill=gray!20, draw=black, thick, rounded corners=1pt] 
    (0,1) -- (0,2) -- (1,2) -- (2,1) -- (2,0) -- (1,0)--cycle; 

 \draw[decorate, decoration={brace, amplitude=5pt}, thick, draw=red]
  (0,0) -- (0,1) node[midway, xshift=-10pt] {\textcolor{red}{$b$}};

\draw[decorate, decoration={brace, mirror, amplitude=6pt}, thick, draw=blue]
  (2,1) -- (2,2) node[midway, xshift=10pt] {\textcolor{blue}{$e$}};

\end{tikzpicture}
    \caption{$\tuplestdbe$ over dense time}
    \label{fig:cover_tdbe_dense}
  \end{subfigure}
  \caption{Alternative representations of the answers to Query $q_3$, in the plane of time per distance}
  \label{fig:compact-subfigs2}
  \vspace{-2em}
\end{figure}

Each of the four compact representations that we study in this article is an alternative way to reduce the number of such pairs.
The first representation, which we call $\tuplest$,
groups these pairs by distance (i.e. horizontally in our figure),
while aggregating time points as intervals.
In this example, this yields three compact answers,
namely 
$([100, 101], 5), $
$([100, 102], 4),$ and
$([101, 102], 3)$.
Each of these is a rectangle with height 1 in our plane,
as illustrated with Figure~\ref{fig:cover_t},
and their union is indeed the area of interest.
The second representation, called $\tuplesd$,
is symmetric, 
in the sense that it
groups results by time points and aggregates distances,
so that compact answers are now rectangle with width 1,
as illustrated with Figure~\ref{fig:cover_d}.
The first representation $\tuplest$ ensures that the number of compact answers is independent of the size of the intervals present in the graph.
However, this number is linear in the size of the intervals used in the query (via the temporal navigation operator).
And conversely for $\tuplesd$.

A third natural attempt (which we call $\tuplestd$) consists in aggregating tuples along both dimensions,
representing our binary relation $R$ as a set of pairs of intervals, one of time points, the other of distances.
These correspond to arbitrary rectangles in our plane.
As can be seen on Figure~\ref{fig:cover_td},
two such rectangles are still needed to cover the whole area (or three if we forbid overlapping rectangles).
Besides, some desirable properties of the first two representations are lost:
the most compact representation (i.e. smallest set of rectangles needed to cover the area) is not unique anymore (e.g. if the area is an "L"-shaped polygon),
and minimizing the number of answers (i.e. finding a minimal set of such rectangles) becomes intractable (if we allow overlapping rectangle).

Moreover, for all three representations,
the number of tuples needed to aggregate all answers to $q_3$ may grow with a change of time granularity.
This can once again be seen in Figure~\ref{fig:all_answers}: 
if we adopt hours rather than days,
then the number of "steps" (bottom left and top right of the area) gets multiplied by 24.
At the limit (i.e. when granularity approaches 0),
we reach dense time,
and the area to cover becomes a rectangle cropped by two lines with slope -1 (drawn in Figure~\ref{fig:cover_tdbe_dense}).
Clearly, such an area cannot be covered by finitely many rectangles.

This is why we introduce a fourth representation, which we call $\tuplestdbe$.
Over dense time, a tuple in $\tuplestdbe$ precisely stands for an area like the one of Figure~\ref{fig:cover_tdbe_dense},
by means of two intervals (one for times, the other for distance),
which intuitively represent the rectangle to be cropped,
together with two values $b$ and $e$
that specify where the two cropping lines intersect respectively with this rectangle.
This representation is also useful over discrete time: in our example, the whole area can be captured with a single tuple that consists of the interval $[100, 102]$ for times,
as illustrated with Figure~\ref{fig:cover_tdbe},
the interval $[3,5]$ for distance, and the values 1 and 1 for $b$ and $e$.
Our most technical result is that binary relations that can be represented as such geometric shapes are closed under composition.
As a consequence,
a TRPQ can only produce a union of such shapes.
So this last format $\tuplestdbe$ overcomes the limitations of the previous ones, in the sense
that answers to a TRPQ can always be represented in a finite way.
This is also the most compact of these four representations.
The price to pay however is an arguably less readable format.
Besides, minimizing a set of tuples under this view is still intractable.

Our main findings (in terms of finiteness, compactness, uniqueness of representation, cost of minimization and query answering) about each of the four representations are summarized at the end of the dedicated section,
in Figure~\ref{fig:summary}.

%



\section{Preliminaries}
\label{sec:preliminaries}

\para{Temporal Graphs}
For maximal generality (and readability),
we use a very simple graph data model,
where each fact is a labelled edge, annotated with a set of time intervals for validity.
Formally, we assume two infinite sets $\N$ of nodes and $\E$ of edge labels.
As is conventional (e.g. in RDF), we represent a fact as a triple $(s, p, o) \in \N \times \E \times \N$
(where $s, p$ and $o$ intuitively stand for "subject", "property" and "object" respectively).
%
In addition, for the temporal dimension of our data,
we assume an underlying temporal domain $\td$ that may be either discrete or dense.
For simplicity, we will use $\integers$ in the former case, and $\rat$ in the latter. 
We also use $\intervals{\td}$ for the set of nonempty intervals over $\td$.

In our model, a database instance simply assigns (finitely many) intervals of validity to (finitely many) triples.
Precisely,
a \emph{temporal graph} $G = \tup{\tdg, \fg,  \valz}$ consists of a bounded \emph{effective temporal domain} $\tdg \in \intervals{\td}$,
a finite set of triples $\fg \subseteq \N \times \E \times \N$, and 
a function $\valz \colon \fg \to 2^{\intervals{\tdg}}$ 
that maps a triple to a finite set of intervals over $\tdg$.
For instance, in Figure~\ref{fig:graph}, $\valz$ assigns the (singleton) set of intervals $\{ [104, 106] \}$ to the triple $(\mathit{Alice}, \mathtt{attends}, \mathit{ISWC})$.




\para{Temporal Regular Path Queries}
\label{sec:language}
We adopt the query language introduced in~\cite{arenas2022temporal}, 
but in a simpler form,
made possible by the simplified data model defined above.
We emphasise that this is without loss of expressivity.

A \emph{Temporal Regular Path Query} (TRPQ) is an expression for the symbol ``$\query$'' in the following grammar:
\begin{align*}
\query  ::=\ &\edge \mid \node \mid \tndelta \mid (\query/\query) \mid  (\query + \query) \mid \query[m,n] \mid \query[m,\_]\\
\edge  ::=\ & \edgelabel \mid \edge^- \\
\node  ::=\ & \pred \mid\ \le k \mid (? \query)  \mid \neg\ \node
\end{align*}
with $\edgelabel \in \E$, $k \in \td, \delta \in \intervals{\td}$, $m, n \in \nn^+$, and $m \leq n$.

Here, $\edge$ and $\node$ are filters on edges and nodes respectively.
The terminal symbol $\pred$ stands for a Boolean predicate that can be evaluated locally for one node $n$,
which we write $n \models \pred$ (for instance, in Figure~\ref{fig:graph},
the node \textit{positive} satisfies the predicate $\mathit{(= positive)}$).
Similarly, the Boolean predicate $\le k$ evaluates whether a time point 
is less than or equal to $k$. The expression
$(? \query)$ filters the nodes that satisfy $\query$, and 
$\neg$ represents logical negation.
The temporal navigation operator 
$\tndelta$ 
stands for navigation in time by any distance in the interval $\delta$,
and the remaining operators are
regular path query (RPQ) operators:
$/$ stands for join, $+$ for union,  
and $\query[m,n]$ for the "repetition" of $\query$ from $m$ to $n$ times. In particular, $\query[0,\_]$ represents Kleene closure (equivalent to the $*$ operator in regular expressions).

The formal semantics of TRPQs is provided in Figure~\ref{fig:semantics},
where $\eval{q}$ is the evaluation of a TRPQ $q$ over a temporal graph $G$.
In this definition,
we use $q^i$ for the TRPQ defined inductively by $q^1 = q$ and $q^{j+1} = q^{j}/q$.
For convenience, we represent (w.l.o.g.) an answer as
two nodes together with a time point and a distance,
rather than two nodes and two time points, i.e. we use tuples of the form 
$\tup{n_1,n_2, t,d}$ rather than
$\Tup{\tup{n_1,t} , \tup{n_2,t+d}}$.
We also use $\nodesg$ to denote the set of nodes that intuitively appear in $G$,
i.e.~all $n$ such that the triple $(n, p, o)$ or $(s, p, n)$ is in $\fg$
for some $s$, $p$ and $o$.
\renewcommand{\arraystretch}{1.2}
\begin{figure}[t]
\[
\begin{array} {rl}
  \eval{\edgelabel}\ = &\{\tup{n_1, n_2, t, 0} \mid t \in \tau \te{for some} \tau \in \val{n_1}{\edgelabel}{n_2}\}\\
  \eval{\edge^-}\ = &\{ \tup{n_2, n_1, t, 0} \mid  \tup{n_1, n_2, t, 0} \in \eval{\edge} \}\\
  \eval{\pred}\ = &\{\tup{n, n, t, 0} \mid n \models \pred \te{and} t \in \tdg \}\\
  \eval{\le k }\ = & \{\tup{n, n, t, 0} \mid n \in \nodesg, t \in \tdg \te{and} t \le k\}\\
  \eval{\tndelta}\ =  &\{\tup{n, n, t, d} \mid n \in \nodesg, t \in \tdg, d \in \delta \te{and} t+d \in \tdg \}\\
  \eval{?\query} = & \{\tup{n,n,t,0} \mid \tup{n,n',t,d} \in \eval{\query}
\te{for some} n' \in \N \te{and} d \in \td\}\\
  \eval{\neg \node}\ = & \big(\{\tup{n,n} \mid n \in \nodesg \} \times \tdg \times \{0\}\big)   \setminus \eval{\node}\\
 \eval{\query_1 / \query_2}\ = &\{\tup{n_1, n_3 , t, d_1 + d_2 } \mid\\
                               & \exists n_2 \colon \tup{n_1, n_2 , t, d_1} \in  \eval{\query_1} \te{and} \tup{n_2, n_3, t + d_1, d_2} \in  \eval{\query_2}\}\\
    \eval{\query_1 + \query_2}\ = &\eval{\query_1} \cup \eval{\query_2}\\
  \eval{\query[m,n]}\ = & \bigcup\limits_{k = m}^n\eval{\query^k}\\
        \eval{\query[m,\_]}\ = &  \bigcup\limits_{k \ge m}\eval{\query^k}
  \end{array} 
\]
\vspace{-1em}
 \caption{Semantics of TRPQs}\label{fig:semantics}
 \vspace{-1em}
\end{figure}
\renewcommand{\arraystretch}{1}


\para{Operations on intervals}
If $\alpha \in \intervals{\integers}$ (resp.~$\intervals{\rat}$) and $d \in \integers$ (resp. $\rat$),
then we use $\alpha + d$ (resp. $\alpha - d$) for the interval $\{t + d \mid t \in \alpha\}$ (resp. $\{t - d \mid t \in \alpha\}$).
Similarly, if $\alpha$ and $\beta$ are intervals,
then we use $\alpha \oplus \beta$ (resp.~$\alpha \ominus \beta$)
for the interval $ \alpha + \{t \mid t \in \beta\}$ (resp.~ $\alpha - \{ t \mid t \in \beta\}$).

\section{Compact answers}
\label{sec:compact}

In this section, we define and study the four compact representations of
answers to a TRPQ sketched in Section~\ref{sec:contributions}.
Let $\tuples$ denote the universe of all tuples that may be output by TRPQs, i.e. 
$\tuples = \N \times \N \times \td \times \td$.
Then each of our four compact representations can be viewed as an alternative format to encode 
subsets of $\tuples$.

We specify each of these four formats as a set of admissible tuples,
denoted
$\tuplest$,
$\tuplesd$,
$\tuplestd$ and
$\tuplestdbe$ respectively.
Let $\tuplesx$ be any of these four sets.
A tuple $\u$ in $\tuplesx$
represents a subset of $\tuples$,
which we call the \emph{unfolding} of $\u$.
And the unfolding of a \emph{set} $U \subseteq \tuplesx$ of such tuples is the union of the unfoldings of the elements of $U$.
We say that $U$
is \emph{compact} if it is finite and if no strictly smaller
(w.r.t.~cardinality) subset of $\tuplesx$ has the same unfolding.
A set $V \subseteq \tuples$ can be \emph{finitely represented} (in $\tuplesx$) if there is a finite $U \subseteq \tuplesx$ with unfolding $V$.

\subsection{Folding time points ($\tuplest$) or distances ($\tuplesd$)} \label{sec:compact_ep}

Each of our two first compact representations aggregates tuples in $\tuples$ along one dimension:
either the time point associated to the first node,
or the distance between times points associated to each node.
The corresponding universes $\tuplest$ and $\tuplesd$ of tuples 
are $\N \times \N \times \intervals{\td} \times \td$ and $\N \times \N \times \td \times \intervals{\td}$ respectively.
The unfolding of a tuple $\tup{n_1, n_2, \tau, d} \in \tuplest$ is 
$\{\tup{n_1, n_2, t, d} \mid t \in \tau\}$,
and analogously  $\{\tup{n_1, n_2, t, d} \mid d \in \delta\}$ for a tuple 
$\tup{n_1, n_2, t, \delta} \in \tuplesd$.

\para{Inductive representation}
In order to understand when the answers $\eval{q}$ to a TRPQ $q$ over a temporal graph $G$ can be finitely represented in $\tuplest$ or $\tuplesd$,
and what the size of such representations may be,
we define by structural induction on $q$ two (not necessarily compact) representation of $\eval{q}$ in $\tuplest$ and $\tuplest$,
noted $\evalcit{q}$ and $\evalcid{q}$ respectively (these representations also pave the way for implementing query evaluation).
For instance, in the case where $q$ is of the form $\pred$, we define 
$\evalcit{\pred}$ as $\{\tup{n, n, \tdg, 0} \mid n \models \pred \}$.

Due to space limitations, 
we only provide the full definition of $\evalcit{q}$ and $\evalcid{q}$ in the appendix (together with proofs of correctness).
We highlight here the least obvious operators.
The first one is the temporal join $\query_1/ \query_2$,
which intuitively composes temporal relations.
If we assume (by induction) that the answers to each operand ($\query_1$ and $\query_2$) are represented in $\tuplest$,
then the representation of $\query_1/ \query_2$ in $\tuplest$ can be computed as a regular join together with simple arithmetic operations on interval boundaries, as follows:
\begin{align*}
 & \evalcit{\query_1/ \query_2}  =  \Big\{\tup{n_1, n_3, ((\tau_1 + d_1) \cap \tau_2)- d_1, d_1 + d_2} \mid \exists n_2 \colon \\
 &                \quad    \tup{n_1, n_2, \tau_1, d_1} \in \evalcit{\query_1}, \tup{n_2, n_3, \tau_2, d_2} \in  \evalcit{\query_2} \te{and} (\tau_1 + d_1) \cap \tau_2 \neq \emptyset\ \Big\}.
\end{align*}
This observation also holds for $\tuplesd$, but with different operations on interval boundaries:
\begin{align*}
 & \evalcid{\query_1/\query_2} = \Big\{\tup{n_1, n_3, t_1, t_2 - t_1 + \delta_2} \mid \exists n_2 \colon \tup{n_1, n_2, t_1 , \delta_1} \in \evalcid{\query_1},  \\
 & \qquad \tup{n_2, n_3, t_2, \delta_2} \in  \evalcid{\query_2} \te{and} t_2 - t_1 \in \delta_1 \Big\}.
\end{align*}
The less obvious of these two operations is $\evalcit{\query_1/ \query_2}$,
which we illustrate with Figure~\ref{fig:tjoin_1}).

Another case of interest is the temporal navigation operator $\tndelta$,
for the second representation $\tuplesd$ (the case of $\tuplest$ is trivial).
Consider a query of the form $q/\tndelta$ (or symmetrically $\tndelta/q$).
If the subquery $\tndelta$ is evaluated independently,
then the output of this subquery may be infinite over dense time, which rules out in practice an inductive evaluation:
\[
  \evalcid{\tndelta} = \{
  \begin{array}[t]{@{}l}
    \tup{n,n,t, ((\delta + t) \cap \tdg) -t } \mid n \in \nodesg, t \in \tdg, (\delta + t) \cap \tdg \neq \emptyset\}.
  \end{array}
\]
However,
if $\evalcid{q}$ is finite,
then the output of the whole query $q/\tndelta$ can be represented finitely,
and effectively computed as follows:
\[
  \evalcid{q/\tndelta} =  \{
  \begin{array}[t]{@{}l}
    \tup{n_1,n_2,t, (\delta' \oplus \delta) \cap \tdg} \mid 
  \tup{n_1,n_2,t, \delta'} \in \evalcid{q}, (t + (\delta' \oplus \delta)) \cap \tdg \neq \emptyset\}.
  \end{array}
\]

\begin{figure}[tbp]
  \footnotesize
  \centering
\begin{tikzpicture}[scale=0.7]

    \draw[|-|, thick, blue] (0,3) -- (3,3) node [midway, above] (TextNode)
        {$\tau_1$};
    \draw[|-|, thick, blue] (4,3) -- (7,3) node [midway, above] (TextNode)
        {$\tau_1 + d_1$};

    \draw[densely dotted] (0,2.3) -- (4,2.3) node [midway, above] (TextNode)
        {$d_1$};

     \draw[|-|, thick, purple] (6,2) -- (8,2) node [midway, above] (TextNode)
        {$\tau_2$};
     \draw[|-|, thick, purple] (9.5,2) -- (11.5,2) node [midway, above] (TextNode)
        {$\tau_2 + d_2$};

    \draw[densely dotted] (6,2.6) -- (9.5,2.6) node [midway, above] (TextNode)
        {$d_2$};

     \draw[|-|, thick, olive] (6,0.6) -- (7,0.6) node [midway, below] (TextNode)
        {$(\tau_1 + d_1) \cap \tau_2$};

    \draw[|-|, thick, violet] (2,0.6) -- (3,0.6) node [midway, below] (TextNode)
        {$((\tau_1 + d_1) \cap \tau_2) -d_1$};

    \draw[|-|, thick, violet] (9.5,0.6) -- (10.5,0.6) node [midway, below] (TextNode)
    {$((\tau_1 + d_1) \cap \tau_2) + d_2$};

    \draw[densely dotted] (2,1.6) -- (9.5,1.6) node [midway, below] (TextNode)
        {$d_1 + d_2$};
    \end{tikzpicture}
    \caption{
Join of two tuples 
 $\tup{n_1, n_2, \tau_1, d_1}$ and $\tup{n_2, n_3, \tau_2, d_2}$ in $\tuplest$.
Each tuple is depicted as two intervals $\tau_i$ and $\tau_i + d_i$.
The two corresponding intervals $\tau_o$ and $\tau_o + d_o$ for the output tuple $\tup{n_1, n_3, \tau_o, d_o}$ are in violet.
}
     \label{fig:tjoin_1}
\end{figure}
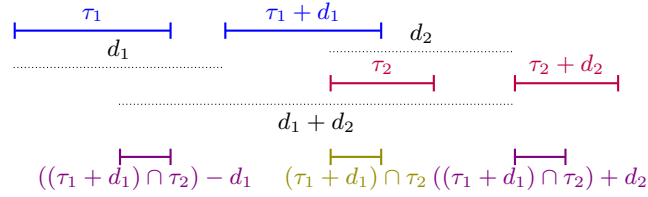
\normalsize
 

\para{Finiteness over dense time}
Over discrete time,
trivially,
$\eval{q}$ can be finitely represented in $\tuples$,\footnote{
 Recall that we assume the effective temporal domain $\tdg$ of $G$ to be bounded.}
therefore as well in any of our four compact representations.
But over dense time, this is not always possible.
For instance, if $q_t$ is the query $\mathit{(= positive)}/\mathtt{tests}^-/\tn{[-7, 0]}$,
and if the graph $G$ of Figure~\ref{fig:graph} is interpreted over dense time,
then $\tup{\mathit{positive},\mathit{Bob}, 112, d} \in \eval{q_t}$ for every rational number $d$ in $[-7, 0]$.
And no tuple in $\tuplest$ can represent more than one of these tuples.
From the definition of $\evalcit{q}$,
the only possible source of non-finiteness for $\tuplest$ is the temporal navigation operator $\tndelta$,
and only if $\delta$ is not a singleton interval (i.e., if it specifies a certain $\emph{range}$ rather than a fixed distance).

Similarly, if $q_d$ is the query $\mathit{(= Bob)}/\mathtt{attends}$,
then $\tup{\mathit{Bob}, \mathit{ISWC}, t, 0} \in \eval{q}$,
for every rational number $t$ in $[102,107]$,
and no tuple in $\tuplesd$ can represent more than one of these.

\para{Compactness}
If $\eval{q}$ can be finitely represented in $\tuplest$ or $\tuplesd$,
then a natural requirement on these representations is conciseness.
It is easy to see that a finite set $U \subseteq \tuplest$ is compact iff all time intervals for the same $n_1, n_2$ and $d$ within $U$ are coalesced.
Formally, let $\sim$ denote the binary relation over $\tuplest$ defined by $\tup{n_1, n_2, \tau_1, d_1} \sim \tup{n_3, n_4, \tau_2, d_2}$ iff $\tup{n_1, n_2, d_1} = \tup{n_3, n_4, d_2}$ and
$\tau_1 \cup \tau_2 \in \intervals{\td}$.
Then $U$ is compact iff $\u_1 \not\sim \u_2$ for all $\u_1, \u_2 \in U$ s.t. $\u_1 \neq \u_2$.
More, there is a unique way to coalesce a finite set of intervals.
Therefore if $V \subseteq \tuples$ can be finitely represented in $\tuplest$,
then $V$ also has a unique compact representation in $\tuplest$.
A symmetric observation holds for a set $U \subseteq \tuplesd$, where we coalesce distance intervals rather than time intervals.

%
Coalescing a set of intervals is known to be in $O(n \log n)$,
and efficient implementations have been devised (see Section~\ref{sec:related-work}).
For this reason, coalescing intermediate results in $\evalcit{q}$ or $\evalcid{q}$ may be an interesting query evaluation strategy.
For instance, this could reduce the size of the operands of a (worst-case quadratic) temporal join.

\para{Size of compact answers}
Our two representations $\tuplest$ and $\tuplesd$ exhibit an interesting symmetry when it comes to the size of compact answers to a query,
intuitively growing with the size of the intervals present in $q$ in the case
of  $\tuplest$,
and the intervals present in $G$ in the case of $\tuplesd$.
This suggests that $\tuplest$ is be better suited to large graphs intervals and small temporal navigation intervals,
and conversely for  $\tuplesd$.

Unfortunately, this does not hold for arbitrary queries.
Indeed, the size of a compact representation of $\eval{q}$ may be affected by the size of the effective temporal domain $\tdg$ alone,
even if all time intervals in the query and graph are singletons,
as soon as $q$ contains an occurrence of the closure operator $[m,\_]$,
as illustrated with the following example:
\begin{example}\label{ex:size_closure}
  Consider a temporal graph $G$ s.t.~$\val{n_1}{e}{n_2} = \{[0,0]\}$ 
and consider the query $q = e/(\text{T}_{[2,2]})[1,\_]$.
The compact representation of $\eval{q}$ is $\{\tup{n_1,n_2,[0,0],d} \mid d \in D\}$ in $\tuplest$,
and $\{\tup{n_1,n_2,0,[d,d]} \mid d \in D\}$ in $\tuplesd$,
where $D = \{ d \in \tdg \cap \nn^+ \mid d \bmod 2 = 0\}$.
\qed
\end{example}
However, for queries without closure operator, which we call \emph{star-free},
the symmetry sketched above holds.
To formalize this, we introduce a notation that we will reuse for the other two representations.
Let $\tuplesx$ be one of $\tuplest$, $\tuplesd$, $\tuplestd$ or $\tuplestdbe$.
Consider a temporal graph $G$ and a star-free query $q$,
such that $\eval{q}$ can be finitely represented in $\tuplesx$.
We fix $G$ and $q$, with the exception of $\tdg$,
and either \ei the intervals present in $G$,
or \eii the intervals present in $q$ (with the requirement that $\eval{q}$ can still be finitely represented in $\tuplesx$).
In case \textit{(i)},
let $n$ be the cumulated length of the intervals in $G$,
and let $V$ be a compact representation of $\eval{q}$ in $\tuplesx$.
We use $\sizea{\tau}{\tuplesx}$ for the function that maps $n$ to the
the number of tuples in $V$.
In case \textit{(ii)},
we use $\sizea{\delta}{\tuplesx}$ with an analogous meaning, but where
$n$ is the cumulated size of the intervals in $q$.

The following results says that the size of the compact representation of $\eval{q}$ in $\tuplest$ (when it exists)
may be affected by the size of the intervals present in $q$,
but not the ones used to label triples in $G$, and conversely for $\tuplesd$:
\begin{restatable}{proposition}{propSizeEvalSingleDimension}\label{prop:size_eval_single_dimension}
In the worst case,
\vspace{-0.5em}
  \begin{align*}
    \sizea{\tau}{\tuplest} &= O(1)  &\sizea{\delta}{\tuplest} & = \Omega(n)\\
    \sizea{\tau}{\tuplesd} &= \Omega(n) &\sizea{\delta}{\tuplesd} & = O(1)
  \end{align*}
\end{restatable} 

\para{Complexity of query answering}
Let $\tuplesx$ be one of $\tuplest$, $\tuplesd$, $\tuplestd$ or $\tuplestdbe$.
We formulate a decision problem
analogous to the classical boolean query answering problem (for atemporal databases),
in such a way that it remains defined even if $\eval{q}$ does not admit a finite representation in $\tuplesx$.
We say that a tuple $\u$ in $\tuplesx$ is a \emph{compact answer} to $q$ over $G$
if its unfolding is a subset of $\eval{q}$ and is maximal among the tuples in $\tuplesx$ that satisfy this condition.
We can now define our (four) problems (where $x$ is either $t, d, td$ or $c$):

\smallskip
\noindent
\begin{center}
\vspace*{-1.5ex}
\framebox[0.7\columnwidth]{%
 \begin{tabular}{l}
   \pbmx\\
   \begin{tabular}{ll}
     \textbf{Input}: & \! temporal graph $G$, TRPQ $q$, tuple $\u \in \tuplesx$\\
     \textbf{Decide}: & $\u$ is a compact answer to $q$ over $G$
   \end{tabular}
 \end{tabular}}
 \vspace*{-2ex}
 \end{center}
 \ \\
  \smallskip

Complexity for these problems is (partly) driven by the size of the input time intervals,
and there is no reason a priori to assume that intervals in the graphs are larger than the ones in the query.
This is why we do not focus on data complexity (where the query is fixed),
but instead on combined complexity,
where the size of the query and data may vary (we leave for future work a finer-grained analysis).
We show in the appendix 
 that the results proven in~\cite{arenas2021temporal} for answering TRPQs in $\tuples$ immediately transfer to $\tuplest$ 
(resp. $\tuplesd$),
even in the case where $\eval{q}$ cannot be finitely represented in $\tuplest$ 
(resp. $\tuplesd$):
\begin{proposition}\ \\ 
     \pbmt and \pbmd are \PSPACE-complete.
\end{proposition}
We also observe that hardness can be proven with a graph of fixed size,
with the exception of the effective temporal domain $\tdg$.
However, the number of operators of the query used in this reduction is not fixed.

\subsection{Folding time points and distances ($\tuplestd$ and $\tuplestdbe$)}\label{sec:foldBoth}
Each of the two compact representations $\tuplest$ and $\tuplesd$ defined above aggregates tuples in $\tuples$ along one dimension (either time points or distances), and both may fail to represent the answers to a query compactly (or even finitely over dense time).
So a first natural attempt to address this limitation consists in aggregating tuples along both dimensions.
Accordingly, we define the universe $\tuplestd$ as $\N \times \N \times \intervals{\td} \times \intervals{\td}$,
and the unfolding of $\tup{n_1, n_2, \tau, \delta} \in \tuplestd$ as $\{\tup{n_1, n_2, t, d} \mid t \in \tau, d \in \delta\}$.

This representation is more compact than the two previous ones,
since unfolding a tuple in $\tuplest$ or $\tuplesd$ yields a subset of some unfolded tuple of $\tuplestd$.
However, as we will see below, minimizing the cardinality of a set of query answers becomes intractable.
Besides, maybe surprisingly,
this new format may also fails to represent query answers in a compact fashion (or finitely over dense time).
To see this,
observe that for two fixed nodes, a tuple in $\tuplestd$ has a natural representation as a rectangle in the Euclidean plane $P$ of times per distances.
The following example shows a query whose output, depicted in Figure~\ref{fig:rectangles}b,
cannot be represented (over dense time) as a finite set of rectangles in $P$.
This example also provides insight about the formalization of our fourth representation (below).

\begin{example}
\label{ex:parallelogram}
Consider a temporal graph $G$ with two edges
such that $\val{n_1}{e_1}{n_2} = \{[0,2]\}$ and $\val{n_2}{e_2}{n_3} = \{[1,3]\}$ 
%
Let $q$ be the query $e_1/\tn{[0,2]}/e_2$.
Then $\eval{q} = \{\tup{n_1,n_2,t,d} \mid t \in [0,2], d \in [0,2] \te{and} t+d \in [1,3] \}$.
So for an answer $\tup{n_1,n_2,t,d} \in \eval{q}$, we have
\begin{align*}
1 \le t + d \le 3, \quad \text{~that is,} \quad 1 - t \le d \le 3 - t. 
\end{align*}
Besides, from the query $q$, we get
$0 \le d \le 2$.
Therefore
\[\max(1 - t, 0) \le d \le \min(3 - t, 2).\]
This observation gives us an interval 
\[\delta_t = [ \max(1 - t, 0), \min(3 - t, 2)]\]
of admissible distances for each $t \in [0,2]$,
so that
$\eval{q} = \{\tup{n_1,n_2,t,d} \mid t \in [0,2] \te{and} d \in \delta_t \}$.
%
Extending this kind of reasoning to the general case,
we derive the formula~\eqref{eq:fourth-repr-forumula} presented below for $\delta_t$.
\qed
\end{example}
As shown in Figure~\ref{fig:rectangles}b,
the area covered by the output of this query can be viewed as a rectangle cropped by two parallel lines,
each with slope -1.
These cropped rectangles turn out to have an essential property:
if two temporal relations have such a shape,
then their composition also does.
As a consequence, associating each tuple with such a cropped rectangle (as opposed to a regular rectangle in $\tuplestd$)
allows us to solve our initial problem:
compute inductively the output of a TRPQ $q$ over a graph $G$
in such a way that the number of output tuples remains independent of the size of the intervals present in either $q$ or $G$.
To represent these cropped rectangles, 
we introduce a fourth format,
which we call $\tuplestdbe$ (where "c" stands for "cropped"),
where each tuple carries,
in addition to a time interval $\tau$ and a distance interval $\delta$ (a.k.a. a rectangle),
the two values $b,e \in \td$ depicted in Figure~\ref{fig:rectangles}b,
which intuitively indicate where the cropping lines intersect the rectangle induced by $\tau$ and $\delta$.
For each time point $t \in \td$,
these define a specific range of distances $\delta_t \subseteq \delta$,
as
\begin{align}
\delta_t =\ \ld{\delta}\ b_\delta + \max(0, b - t) \ ,\ e_\delta - \max(0, t - e)\ \rd{\delta}
\label{eq:fourth-repr-forumula}
\end{align}
where
$\ld{\delta}$ and $\rd{\delta}$ stand for the left and right delimiters of $\delta$,
and $b_\delta$ and $e_\delta$ for its left and right boundaries
For instance, if $\delta = [2,6)$,
then $\ld{\delta}$ is ``['',
$b_\delta$ is 2,
$e_\delta$ is 6
and $\rd{\delta}$ is ``)''.

Accordingly, a tuple $\tup{n_1, n_2, \tau, \delta, b, e} \in \N \times \N \times \intervals{\td} \times \intervals{\td} \times \td \times \td$ is in
$\tuplestdbe$ iff $\delta_t$ is nonempty for every $t \in \tau$.
And the unfolding of this tuple is $\{\tup{n_1, n_2, t, d} \mid t \in \tau, d \in \delta_t\}$.

\para{Inductive representation}
As we did for our first two representations,
we define (in appendix) by structural induction on $q$ two representations $\evalcitd{q}$ and $\evalcitdbe{q}$ of $\eval{q}$ in $\tuplestd$ and $\tuplestdbe$ respectively.
Our most technical result is correctness of this definition for the join operator in $\tuplestdbe$
(and to a lesser extent the temporal navigation operator).

We reproduce these definitions here to make apparent the fact these can be computed by means of simple arithmetic operations on time points and interval boundaries (together with a regular join on nodes for the join operator).
The inductive evaluation $\evalcitdbe{\query_1/\query_2}$ in $\tuplestdbe$ of the query $\query_1/\query_2$ is defined as: 
\[\{\u_1 \tjoin \u_2 \mid \u_1 \in \evalcitdbe{\query_1}, \u_2 \in \evalcitdbe{\query_2}, \u_1 \sim \u_2\}\]
where $\u_1 \sim \u_2$ and $\u_1 \tjoin \u_2$ are defined as follows.\\
For $\u_1 = \tup{n_1, n_2, \tau_1, \delta_1, b_1, e_1}$ and $\u_2 = \tup{n_3, n_4, \tau_2, \delta_2, b_2, e_2}$, let
\begin{align*}
\delta'_1 =&\ \ld{\delta_1}\ b_{\delta_1} + \max(0, b_1 - b_{\tau_1}), e_{\delta_1} - \max(0, e_{\tau_1} - e_1)\ \rd{\delta_1}, \te{and}\\
\tau = &  (((\tau_1 \oplus \delta'_1) \cap \tau_2) \ominus \delta'_1) \cap \tau_1.
\end{align*}
Then the relation $\sim\ \subseteq \tuplestdbe \times \tuplestdbe$ is defined as $\u_1 \sim \u_2$ iff $n_2 = n_3$ and $\tau \neq \emptyset$.\\
And if $\u_1 \sim \u_2$, then $\u_1 \tjoin \u_2$ is defined as  $\tup{n_1, n_4, \tau, \delta_1 \oplus \delta_2, b,e}$,
with $b = \max(b_1, b_2 - b_{\delta_1})$ and $e = \min(e_1, e_2 - e_{\delta_1})$.

Similarly, if $b_{\tdg}$ and $e_{\tdg}$ are the boundaries of the interval $\tdg$,
then $\evalcitdbe{\tndelta}$ is defined as
\[\{\tup{n,n,\tdg, \delta,b_{\tdg},e_{\tdg}} \tjoin \tup{n,n,\tdg, [0,0],b_{\tdg},e_{\tdg}} \mid n \in \nodesg \}.\]

\para{Finiteness over dense time}
We already illustrated in Figure~\ref{fig:rectangles}b why $\eval{q}$ may not be finitely representable in $\tuplestd$.
In contrast, $\eval{q}$ can be finitely represented in $\tuplestdbe$,
as a direct consequence of the (correctness of the) inductive definition of $\evalcitdbe{q}$,
which produces finitely many tuples (notably, the operator $\query_1/\query_2$ produces at most one tuple per pair $(\u_1, \u_2) \in \evalcitdbe{\query_1} \times \evalcitdbe{\query_2}$).

\begin{figure}[t]
  \begin{subfigure}[t]{0.7\textwidth}
    \centering
    \begin{minipage}[t]{0.3\textwidth}
        \begin{tikzpicture}[scale=\scaleExamplesSecondFig]

  \foreach \x in {0,1,2} {
    \draw[gray!50] (\x, 0) -- (\x, 2);
  }
  \foreach \y in {0,1,2} {
    \draw[gray!50] (0, \y) -- (2, \y);
  }

  \foreach \i in {0,1,2} {
    \node[below] at (\i, 0) {\scriptsize \i};
  }

  \foreach \j in {0,1,2} {
    \node[left] at (0, \j) {\scriptsize \j};
  }

  \draw[->, thick] (0, 0) -- (2.3, 0) node[right] {$t$};
  \draw[->, thick] (0, 0) -- (0, 2.3) node[above] {$d$};

  \path[fill=gray!20, thick, draw=black, rounded corners=\roundingEdgesExamples] 
    (0,1) rectangle (1,2);  

  \path[fill=gray!20, thick, draw=black, rounded corners=\roundingEdgesExamples] 
    (0,0) rectangle (2,1);  

\end{tikzpicture}
    \end{minipage}%
    \hspace*{3ex}
    \begin{minipage}[t]{0.3\textwidth}
        \begin{tikzpicture}[scale=\scaleExamplesSecondFig]

  \foreach \x in {0,1,2} {
    \draw[gray!50] (\x, 0) -- (\x, 2);
  }
  \foreach \y in {0,1,2} {
    \draw[gray!50] (0, \y) -- (2, \y);
  }

  \foreach \i in {0,1,2} {
    \node[below] at (\i, 0) {\scriptsize \i};
  }

  \foreach \j in {0,1,2} {
    \node[left] at (0, \j) {\scriptsize \j};
  }

  \draw[->, thick] (0, 0) -- (2.3, 0) node[right] {$t$};
  \draw[->, thick] (0, 0) -- (0, 2.3) node[above] {$d$};

  \path[fill=gray!20, thick, draw=black, rounded corners=\roundingEdgesExamples] 
    (0,0) rectangle (1,2);  

  \path[fill=gray!20, thick, draw=black, rounded corners=\roundingEdgesExamples] 
    (1,0) rectangle (2,1);  

\end{tikzpicture}
    \end{minipage}
    \label{fig:xxx}
  \end{subfigure}
  \begin{subfigure}[t]{0.25\textwidth}
    \centering




\begin{tikzpicture}[scale=\scaleExamplesSecondFig]

  \foreach \x in {0,1,2} {
    \draw[gray!50] (\x, 0) -- (\x, 2);
  }
  \foreach \y in {0,1,2} {
    \draw[gray!50] (0, \y) -- (2, \y);
  }

  \foreach \i in {0,1,2} {
    \node[below] at (\i, 0) {\scriptsize \i};
  }

  \foreach \j in {0,1,2} {
    \node[left] at (0, \j) {\scriptsize \j};
  }

  \draw[->, thick] (0, 0) -- (2.3, 0) node[right] {$t$};
  \draw[->, thick] (0, 0) -- (0, 2.3) node[above] {$d$};

  \draw[fill=gray!20, draw=black, thick, rounded corners=1pt] 
    (0,1) -- (0,2) -- (1,2) -- (2,1) -- (2,0) -- (1,0)--cycle; 

  \draw[thick, red] (-1,2) -- (2,-1);         
  \draw[thick, blue] (0,3) -- (3,0);        

  \draw[decorate, decoration={brace, amplitude=6pt}, thick, draw=red]
  (0,0) -- (0,1) node[midway, xshift=-10pt] {\textcolor{red}{$b$}};

\draw[decorate, decoration={brace, mirror, amplitude=6pt}, thick, draw=blue]
  (2,1) -- (2,2) node[midway, xshift=10pt] {\textcolor{blue}{$e$}};

\end{tikzpicture}

    \label{fig:ex2}
  \end{subfigure}
  \vspace*{-1em}
  \caption{\ %
  \textbf{(a)} Two minimal covers in $\tuplestd$
  \qquad\qquad
  \textbf{(b)} Answers to the query of Example~\ref{ex:parallelogram} }
\vspace*{-1em}
  \label{fig:rectangles}
\end{figure}

\para{Compactness}
Let $U \subseteq \tuplestd$ and $V \subseteq \tuples$ be two sets of tuples that share the same nodes $n_1$ and $n_2$.
Then $U$ unfolds as $V$ iff they intuitively cover the same area in the Euclidean plane of times per distances,
i.e. if
\[\bigcup \{\tau \times \delta \mid \tup{n_1,n_2,\tau,\delta} \in U\} = \{(t,d) \mid \tup{n_1,n_2,t,d} \in V\}.\]

There may be several compact representations in $\tuplestd$ for the same $V \subseteq \tuples$.
For instance, the minimal number of rectangle needed to cover an ``L''-shaped polygon is two,
and there are several such covers,
as illustrated with Figure~\ref{fig:rectangles}a.
This argument easily generalizes to discrete time.

Besides, minimizing the representation of $\eval{q}$ in $\tuplestd$ 
(i.e. computing a minimal set of tuples in $\tuplestd$ with unfolding $\eval{q}$)
is intractable.
To see this, we first observe that computing such a set is harder than deciding whether there exists one with size $k$ (for a given $k$).
Next, the following problem is known to be NP-complete: given a rectilinear polygon $r$ and a number $k$, decide whether there is a set of at most $k$ (possibly overlapping) rectangles that exactly cover $r$~\cite{culberson1994covering,aupperle1988covering}.
This problem reduces to ours, observing that for any rectilinear polygon $r$,
a query $q$ (with only unions) and graph $G$ can be constructed in polynomial time so that $\eval{q}$ covers exactly $r$.
This hardness result also immediately translates to $\tuplestdbe$:
a rectangle is a specific case of a cropped rectangle, therefore the same reduction can be used,
and if there is a minimization with at most $k$ tuples in $\tuplestd$, then there is also one in $\tuplestdbe$.

An important difference though between these two representations and the two previous ones ($\tuplest$ and $\tuplesd$) is that a compact set $U$ of answers in $\tuplestd$ or $\tuplestdbe$ may be \emph{redundant},
meaning that two tuples $\u_1$ and $\u_2$ in $U$ may have overlapping unfoldings.
This can for instance be seen in Figure~\ref{fig:cover_td}, for $\tuplestd$.
If we require tuples to be non-redundant,
then the representations of answers to a query is in general less concise.
But it may be better suited for downstream tasks, such as aggregation.
Besides, tractability of minimization in $\tuplestd$ is regained,
because finding a cover with a minimal number of non-overlapping rectangles is tractable \cite{Keil1986MinimallyCA,DBLP:conf/wg/Eppstein09}.
However, uniqueness is not regained,
as shown in Fig.~\ref{fig:rectangles}a.



 \para{Size of compact answers}
If $V \subseteq \tuples$ can be finitely represented in $\tuplestd$,
then trivially, a compact representation of $V$ in $\tuplestd$ must be smaller than the compact representation of $V$ in $\tuplest$ (resp. $\tuplesd$),
if the latter exists.
So compact answers under this representation must be smaller than under the two previous ones.
However, 
the size of a compact representation may still be affected by the size of time intervals in $q$ (for a star-free $q$ already).
In contrast,
for $\tuplestdbe$,
immediately from the definition of $\evalcitdbe{q}$,
the number of tuples in $\evalcitdbe{q}$ (for a star-free $q$) is independent of the size of the intervals in $G$ or $q$ (even though $\evalcitdbe{q}$ is not necessarily compact).
\begin{restatable}{proposition}{propSizeEvalBothDimensions}\label{prop:size_eval_both_dimensions}
In the worst case,
\vspace{-0.5em}
  \begin{align*} 
    \sizea{\tau}{\tuplestd} & = O(1) & \sizea{\delta}{\tuplestd} & = \Omega(n)\\
    \sizea{\tau}{\tuplestdbe} & = O(1) & \sizea{\delta}{\tuplestdbe} & = O(1)
  \end{align*}
\end{restatable} 

\para{Complexity of query answering}
The hardness results of~\cite{arenas2021temporal} over $\tuples$ can also be lifted to $\tuplestd$ and $\tuplestdbe$,
and this bound is tight for $\tuplestd$:
\begin{proposition}\ \\
  \pbmtd is \PSPACE-complete, 
  \pbmtdbe is \PSPACE-hard.
\end{proposition}

\begin{figure*}[t!]
         \centering
       \begin{tabular}{lccccccccc}
  \toprule
  \multirow{2}{*}{} &
  \multirow{2}{*}{Finite} &
  \multirow{2}{*}{Unique} &
  \multicolumn{2}{c}{Size (star-free $q$)}&
  \multirow{2}{*}{Minimization}&
  Query \\
  \cmidrule(lr){4-5}
  &(dense time)&&graph intervals& query intervals&&answering\\
                                                                                  
                                        

  \midrule
$\tuplest$ & no & yes & $O(1)$& $\Omega(n)$ & $O(n \log n)$& \PSPACE-c & \\
$\tuplesd$& no & yes & $\Omega(n)$& $O(1)$ & $O(n \log n)$&\PSPACE-c & \\
  $\tuplestd$ &{no}&{no}&{$O(1)$}&{$\Omega(n)$}& \NP-h / $O(n^{2.5})$ &\PSPACE-c &\\
$\tuplestdbe$ & yes & no &  $O(1)$& $O(1)$ & \NP-h & \PSPACE-h& \\
\bottomrule
\end{tabular}

  \caption{
           Summary of results.
 }
\label{fig:summary}
\vspace{-1em}
\end{figure*}


\section{Related work}
\label{sec:related-work}

In temporal relational databases, tuples (or attributes) are most commonly
associated with a \emph{single} time interval that represents either validity or transaction time~\cite{DBLP:conf/ebiss/BohlenDGJ17}.
Intervals are commonly used instead of time points as a compact representation.
To maintain a compact and unique representation through operations,
the \emph{coalescing operator}, which merges value-equivalent tuples over consecutive or overlapping time intervals, has received a lot of attention.
B\"ohlen et al.~\cite{bohlen2026coalescing} showed that coalescing can be implemented in SQL, and provided a comprehensive analysis of various coalescing algorithms.
Al-Kateb et al.~\cite{al2005cme} investigated coalescing in the attribute timestamped CME temporal relational model and the work of~\cite{DBLP:journals/pvldb/DignosGNGB19} defines coalescing for temporal multi-set relations. Zhou et al.~\cite{zhou2006efficient} exploited SQL:2003's analytical functions for the computation of coalescing, which to date, is the state-of-the-art technique and can be computed efficiently in  $\mathcal{O}(n\log n)$. In our first two representations $\tuplest$ and $\tuplesd$, which use a single interval, coalescing can be used to achieve a compact representation.
However, this is not applicable for $\tuplestd$ and $\tuplestdbe$.
Bitemporal databases~\cite{DBLP:reference/db/JensenS18c} use two intervals (or rectangles) as a compact representation for two time points,
typically one for validity and one for transaction times.
Our third representation $\tuplestd$ also uses two intervals.
but both are used for (a generalized form of) validity (we are not modeling transaction times).
In bitemporal databases,
coalescing does not provide a unique representation~\cite{DBLP:conf/dagstuhl/Toman97} and,
since it is designed as a non-redundant operation,
may not provide the most compact representation.

Also relevant to our work is the efficient computation of temporal joins over intervals.
There has been a long line of research on temporal joins~\cite{DBLP:journals/vldb/GaoJSS05}, ranging from partition-based~\cite{DBLP:conf/sigmod/DignosBG14,DBLP:journals/vldb/CafagnaB17}, index-based~\cite{DBLP:conf/sigmod/EnderleHS04,DBLP:conf/sigmod/KaufmannMVFKFM13}, and sorting based~\cite{DBLP:conf/icde/PiatovHD16,DBLP:journals/vldb/BourosMTT21} techniques.
Recently, in~\cite{DBLP:journals/vldb/DignosBGJM22} it has been shown that a temporal join with the overlap predicate can be transformed into a sequence of two range joins. 
Our inductive representations of answers require temporal joins and range joins (e.g. for $\tuplest$ and $\tuplesd$ respectively).


Finally, this paper builds upon the original proposal made in~\cite{arenas2022temporal} to extend regular path queries over property graphs with a temporal navigation operator—effectively allowing selection, join, and union of binary temporal relations. Our work focuses on producing compact answers that are finite even in the case of continuous time, a problem that was left open in~\cite{arenas2022temporal}.

\section{Conclusions}
\label{sec:conclusions}




We investigated how to compactly represent answers to queries over binary temporal relations in the TRPQ setting, where both data and queries may include time intervals for validity and temporal navigation, respectively. To our knowledge, this question was previously open.
We defined and analyzed four alternative representations of compact answers to TRPQs, varying in conciseness and potential use.
Notably, the fourth representation ensures that query answers are always finitely representable,
and that their number is independent of the length of the input (graph and query) intervals.
We see this as a useful step towards integrating temporal navigation into database systems.

An open question is whether non-redundancy is tractable under this fourth representation, in particular whether tractability for minimizing compact answers is regained.


\begin{credits}
\subsubsection{\ackname}
This research has been partially supported
by the HEU project CyclOps (GA n.~101135513),
by the Province of Bolzano and FWF through project OnTeGra (DOI 10.55776/PIN8884924).
by the Province of Bolzano and EU through projects
ERDF-FESR~1078 CRIMA, and
ERDF-FESR~1047 AI-Lab,
by MUR through the PRIN project 2022XERWK9 S-PIC4CHU,
and
by the EU and MUR through the PNRR project PE0000013-FAIR.
\end{credits}

\newpage

\bibliographystyle{splncs04}
\bibliography{main-bib}

\begin{thebibliography}{10}
\providecommand{\url}[1]{\texttt{#1}}
\providecommand{\urlprefix}{URL }
\providecommand{\doi}[1]{https://doi.org/#1}

\bibitem{al2005cme}
Al-Kateb, M., Mansour, E., El-Sharkawi, M.E.: {CME}: {A} temporal relational model for efficient coalescing. In: Proc.\ of the 12th Int.\ Symp.\ on Temporal Representation and Reasoning (TIME). pp. 83--90. IEEE Computer Society (2005)

\bibitem{arenas2022temporal}
Arenas, M., Bahamondes, P., Aghasadeghi, A., Stoyanovich, J.: Temporal regular path queries. In: Proc.\ of the 38th IEEE Int.\ Conf.\ on Data Engineering (ICDE). pp. 2412--2425. IEEE Computer Society (2022)

\bibitem{arenas2021temporal}
Arenas, M., Bahamondes, P., Stoyanovich, J.: Temporal regular path queries: Syntax, semantics, and complexity. CoRR Technical Report arXiv:2107.01241, arXiv.org e-Print archive (2021), \url{https://arxiv.org/abs/2107.01241}, available at \protect\url{https://arxiv.org/abs/2107.01241}

\bibitem{aupperle1988covering}
Aupperle, L., Conn, H.E., Keil, J.M., O'Rourke, J.: Covering Orthogonal Polygons with Squares. Johns Hopkins University, Department of Computer Science (1988)

\bibitem{DBLP:conf/ebiss/BohlenDGJ17}
B{\"o}hlen, M.H., Dign{\"o}s, A., Gamper, J., Jensen, C.S.: Temporal data management -- {An} overview. In: Tutorial Lectures of the 7th European Summer School on Business Intelligence and Big Data (eBISS). Lecture Notes in Business Information Processing, vol.~324, pp. 51--83. Springer (2017)

\bibitem{bohlen2026coalescing}
B{\"o}hlen, M.H., Snodgrass, R.T., Soo, M.D.: Coalescing in temporal databases. In: Proc.\ of the 22nd Int.\ Conf.\ on Very Large Data Bases (VLDB). pp. 180--191 (1996)

\bibitem{DBLP:journals/vldb/BourosMTT21}
Bouros, P., Mamoulis, N., Tsitsigkos, D., Terrovitis, M.: In-memory interval joins. Very Large Database J.  \textbf{30}(4),  667--691 (2021)

\bibitem{DBLP:journals/vldb/CafagnaB17}
Cafagna, F., B{\"o}hlen, M.H.: Disjoint interval partitioning. Very Large Database J.  \textbf{26}(3),  447--466 (2017)

\bibitem{culberson1994covering}
Culberson, J.C., Reckhow, R.A.: Covering polygons is hard. J.\ of Algorithms  \textbf{17}(1),  2--44 (1994)

\bibitem{DBLP:conf/sigmod/DignosBG14}
Dign{\"o}s, A., B{\"o}hlen, M.H., Gamper, J.: Overlap interval partition join. In: Proc.\ of the 35th ACM Int.\ Conf.\ on Management of Data (SIGMOD). pp. 1459--1470 (2014)

\bibitem{DBLP:journals/vldb/DignosBGJM22}
Dign{\"{o}}s, A., B{\"{o}}hlen, M.H., Gamper, J., Jensen, C.S., Moser, P.: Leveraging range joins for the computation of overlap joins. Very Large Database J.  \textbf{31}(1),  75--99 (2022)

\bibitem{DBLP:journals/pvldb/DignosGNGB19}
Dign{\"o}s, A., Glavic, B., Niu, X., Gamper, J., B{\"o}hlen, M.H.: Snapshot semantics for temporal multiset relations. Proc.\ of the VLDB Endowment  \textbf{12}(6),  639--652 (2019)

\bibitem{DBLP:conf/sigmod/EnderleHS04}
Enderle, J., Hampel, M., Seidl, T.: Joining interval data in relational databases. In: Proc.\ of the 25th ACM Int.\ Conf.\ on Management of Data (SIGMOD). pp. 683--694 (2004)

\bibitem{DBLP:conf/wg/Eppstein09}
Eppstein, D.: Graph-theoretic solutions to computational geometry problems. In: Revised Papers the 35th Int.\ Workshop on Graph-Theoretic Concepts in Computer Science (WG). Lecture Notes in Computer Science, vol.~5911, pp. 1--16 (2009)

\bibitem{francis2018cypher}
Francis, N., Green, A., Guagliardo, P., Libkin, L., Lindaaker, T., Marsault, V., Plantikow, S., Rydberg, M., Selmer, P., Taylor, A.: {Cypher}: An evolving query language for property graphs. In: Proc.\ of the 39th ACM Int.\ Conf.\ on Management of Data (SIGMOD). pp. 1433--1445 (2018)

\bibitem{DBLP:journals/vldb/GaoJSS05}
Gao, D., Jensen, C.S., Snodgrass, R.T., Soo, M.D.: Join operations in temporal databases. Very Large Database J.  \textbf{14}(1),  2--29 (2005)

\bibitem{harris2013sparql}
Harris, S., Seaborne, A.: {SPARQL}~1.1 query language. {W3C} {Recommendation}, World Wide Web Consortium (Mar 2013), available at \protect\url{http://www.w3.org/TR/sparql11-query}

\bibitem{DBLP:reference/db/JensenS18c}
Jensen, C.S., Snodgrass, R.T.: Bitemporal relation. In: Encyclopedia of Database Systems. Springer, 2nd edn. (2018)

\bibitem{DBLP:conf/sigmod/KaufmannMVFKFM13}
Kaufmann, M., Manjili, A.A., Vagenas, P., Fischer, P.M., Kossmann, D., F{\"a}rber, F., May, N.: Timeline index: a unified data structure for processing queries on temporal data in {SAP} {HANA}. In: Proc.\ of the 34th ACM Int.\ Conf.\ on Management of Data (SIGMOD). pp. 1173--1184 (2013)

\bibitem{Keil1986MinimallyCA}
Keil, J.M.: Minimally covering a horizontally convex orthogonal polygon. In: Proc.\ of the 2nd Annual ACM SIGACT/SIGGRAPH Symposium on Computational Geometry (SCG). pp. 43--51 (1986)

\bibitem{DBLP:conf/icde/PiatovHD16}
Piatov, D., Helmer, S., Dign{\"o}s, A.: An interval join optimized for modern hardware. In: Proc.\ of the 32th IEEE Int.\ Conf.\ on Data Engineering (ICDE). pp. 1098--1109. IEEE Computer Society (2016)

\bibitem{DBLP:conf/dagstuhl/Toman97}
Toman, D.: Point-based temporal extensions of {SQL} and their efficient implementation. In: Temporal Databases, Dagstuhl. Lecture Notes in Computer Science, vol.~1399, pp. 211--237. Springer (1997)

\bibitem{zhou2006efficient}
Zhou, X., Wang, F., Zaniolo, C.: Efficient temporal coalescing query support in relational database systems. In: Proc.\ of the 17th Int.\ Conf.\ on Database and Expert Systems Applications (DEXA). Lecture Notes in Computer Science, vol.~4080, pp. 676--686. Springer (2006)

\end{thebibliography}

\newpage

\appendix
\section{Appendix}



The structure adopted in this appendix differs from one followed in the article.
Results here are grouped by topic (inductive representation, complexity, etc.), rather than by format ($\tuplest$, $\tuplesd$, etc.).
This allows us to factorize some proofs, and emphasize what differs from one case to the other.

The most technical results are the correctness of the inductive representation of compact answers in $\tuplestdbe$ (in particular for the join operator), proven in Section~\ref{sec:evalbe_correct},
and to a lesser extent the analogous result for  $\tuplestd$, proven in Section~\ref{sec:evaltd_correct}.

On the other hand, complexity proofs (in Section~\ref{sec:complexity}) leverage results already proven in~\cite{arenas2021temporal}.

Results pertaining to the size of compact answers (in Section~\ref{sec:size}),
follow either from the corresponding inductive representations (for the upper bounds), or from simple examples (for the lower bounds), similar to the ones already provided in the body of the article.

Almost all arguments that pertain to compactness and cost of coalescing answers are already provided in Section~\ref{sec:compact}, so we complete these for a single case, in Section~\ref{sec:minimization}.

For finiteness over dense time, all negative results (i.e.~non-finiteness) are illustrated in the article, and all positive results (i.e.~finiteness) follow from the definitions of the inductive representations (and the fact that these are correct).


\subsection{Notation}
\label{sec:notation}

\para{Sets, relations, order}
We use $\dom{R}$ and $\range{R}$ for the domain and range of a binary relation, respectively.
For a set $S$ and a (possibly partial) order $\preceq$ over $S$,
we denote with $\max_\preceq S$
the set of maximal elements in $S$ w.r.t.\ $\preceq$,
i.e., $\{s \in S \mid s \preceq s' \te{implies} s = s' \te{for all} s' \in S\}$.

\para{Complement of a set of intervals}
Let $\alpha$ be a bounded interval in $\intervals{\tdg}$,
and let $S$ be a finite set of intervals s.t. $\bigcup S \subseteq \alpha$.
We us $\compl{S}{\alpha}$ to denote the complement of $\bigcup S$ in $\alpha$
represented as maximal intervals, i.e.~:
\[\compl{S}{\alpha} = \max_\subseteq \{ I \subseteq \bigcup S \setminus \alpha \mid I \in \intervals{\td} \}\]

\para{Interval of distances for a given time point $t$}
Let $\tup{n_1, n_2, \tau, \delta, b, e} \in \tuplestdbe$, and let $t \in \tau$.\\
In the article,
we defined the interval $\delta_t$ for each $t$
as 
\[\ld{\delta}\ b_{\delta} + \max(0, b - t) \ ,\ e_{\delta_i} - \max(0, t - e)\ \rd{\delta}\]
In this appendix, we will use $\delta(t)$ instead of $\delta_t$.\\
This notation will allow us to write $\delta_1(t)$ when several tuples are involved.\\
Note that the time points $b$ and $e$ in this notaton are still omitted, for conciseness,
because they are clear from the context.


\subsection{Inductive characterizations}
\label{sec:inductive}

Let $q$ be a TRPQ and $G$ a temporal graph.\\
Then $\eval{q}$ is the set of anwers to $q$ over $G$ (represented as tuples in $\tuples$).

\nt In this section, we provide the full definition of the four inductive representations of $\eval{q}$ discussed in the article,
in
$\tuplest$, 
$\tuplesd$, 
$\tuplestd$ and
$\tuplestdbe$ respectively,
and prove that they are correct.\\
These representations are denoted as
$\evalcit{q}$,
$\evalcid{q}$,
$\evalcitd{q}$ and
$\evalcitdbe{q}$ respectively.

\subsubsection{In $\tuplest$}
\label{sec:evalt}
\paragraph{Definition.}
\label{sec:evalt_def}
If $q$ is a TRPQ and $G = \tup{\tdg, \fg,  \valz}$ a temporal graph, then the representation $\evalcit{q}$ of $\eval{q}$ in $\tuplest$ is defined inductively as follows:
\renewcommand{\arraystretch}{1.4}
\[
  \begin{array}{rl}
    \evalcit{\edgelabel} =& \{\tup{n_1, n_2, \tau, 0} \mid \tau \in \val{n_1}{\edgelabel}{n_2} \}\\
     \eval{\edge^-}\ = &\{ \tup{n_2, n_1, \tau, 0} \mid  \tup{n_1, n_2, \tau, 0} \in \eval{\edge} \}\\
      \evalcit{\pred} =& \{\tup{n, n, \tdg, 0} \mid n \models \pred \}\\
    \evalcit{\le k} =&
                     \begin{cases}
                       \{\tup{n, n, \ld{\tdg} b_{\tdg}, k], 0} \mid n \in \nodesg \} \te{if} k \in \tdg, \te{and}\\
                       \emptyset \te{otherwise}
                     \end{cases}\\
    \evalcit{\tndelta} =& \{\tup{n, n, [t_1,t_1], t_2 - t_1} \mid n \in \nodesg, t_1 \in \tdg \te{and} t_2 \in (\delta + t_1) \cap \tdg\}\\
  \evalcit{(?\query)} = & \{\tup{n,n, \tau, 0} \mid  \tup{n, n', \tau, d} \in \evalcit{\query} \te{for some} n' \in \nodesg, d \in \td \}\\
  \evalcit{\neg \node} =& \bigcup\limits_{n \in \nodesg}\Big\{\tup{n,n,\tau, 0} \mid \tau \in \compl{\{\tau' \mid \tup{n,n,\tau', 0} \in \evalcit{\node}\}}{\tdg}\Big\}\\
  \evalcit{\query_1/ \query_2}  = &  \Big\{\tup{n_1, n_3, ((\tau_1 + d_1) \cap \tau_2)- d_1, d_1 + d_2} \mid \exists n_2 \colon \\
 &                \tup{n_1, n_2, \tau_1, d_1} \in \evalcit{\query_1}, \tup{n_2, n_3, \tau_2, d_2} \in  \evalcit{\query_2} \te{and} (\tau_1 + d_1) \cap \tau_2 \neq \emptyset\ \Big\}.\\

  \evalcit{\query_1 + \query_2} = & \evalcit{\query_1} \cup \evalcit{\query_2}\\
  \evalcit{\query[m,n]} = & \bigcup\limits_{k = m}^n\evalcit{\query^k}\\
  \evalcit{\query[m,\_]} = & \bigcup\limits_{k \ge m}\evalcit{\query^k}
\end{array} 
\]
\renewcommand{\arraystretch}{1}

We observe that when $q$ is of the form
($\query_1 + \query_2$), ($\query[m,\_]$) and ($\query[m,n]$),
the defnition of
$\evalcit{q}$ is nearly identical to the one of $\eval{q}$.\\
This also holds for the three representations below.

\paragraph{Correctness.}
\label{sec:evalt_correct}
We start with a (trivial) lemma:
\begin{restatable}{lemma}{lemma:testDistZero}\label{lemma:test_dist_zero}
  Let 
  $G =
  \tup{\tdg, \fg, \valz}$
  be a temporal graph and let $q$ be an expression for the symbol $\node$ or $\edgelabel$ in the grammar of Section~\ref{sec:language}.\\
  Then:
  \begin{itemize}
  \item each tuples in $\eval{q}$ is of the form $\tup{n_1, n_2, t, 0}$ for some $n_1, n_2$ and $t$,
  \item each tuples in $\evalcit{q}$ is of the form $\tup{n_1, n_2, \tau, 0}$ for some $n_1, n_2$ and $\tau$.
  \end{itemize}
\end{restatable}
\begin{proof}
Immediate from the definitions of $\eval{q}$ and $\evalcit{q}$.
\end{proof}

\nt The following result states that the representation $\evalcit{q}$ is correct:
\begin{restatable}{proposition}{correctT}\label{prop:correct_t}
  For a temporal graph $G = \tup{\tdg, \fg, \valz}$ and a TRPQ $q$, the unfolding of $\evalcit{q}$ is $\eval{q}$.
\end{restatable}

\begin{proof}\ \\
  Let $G = \tup{\tdg \fg, \valz}$ be a temporal graph,
  and let $q$ be a TRPQ.\\
  We show below that:
  \begin{enumerate}[(I)]
  \item for any $\tup{n_1,n_2,t, d} \in \eval{q}$,
    there is a $\tau \in \intervals{\td}$ such that\label{item:tcorr_1}
    \begin{enumerate}[(a)]
    \item $\tup{n_1,n_2, \tau, d} \in \evalcit{q}$\label{item:tcorr_1_1}, and
    \item $t \in \tau$,\label{item:tcorr_1_2}
    \end{enumerate}
  \item for any $\tup{n_1,n_2,\tau, d} \in \evalcit{q}$
    for any $t \in \tau$,\\
    \qquad $\tup{n_1,n_2, t, d}$ is in $\eval{q}$.
    \label{item:tcorr_2}
  \end{enumerate}
  
    \nt We proceed by induction on the structure of $q$.\\
    If $q$ is of the form $\edgelabel$, $\edge^-$, $\pred$, $\le k$, $(\query + \query)$, $\query[m,n]$ or $\query[m,\_]$,
    then~\ref{item:tcorr_1} and~\ref{item:tcorr_2} immediately follow from the definitions of
    $\eval{q}$ and 
    $\evalcit{q}$.\\
    So we focus below on the four remaining cases:
  \begin{itemize}
    

  
  \item $q = \tndelta$.\ \\
    From the above definitions, we have:
    \begin{align*}
      \eval{q}\ = &\{\tup{n, n, t, d} \mid n \in \nodesg, t \in \tdg, d \in \delta \te{and} t + d \in \tdg\}\\
  \evalcit{q}\ =&\{\tup{n, n, [t_1,t_1], t_2 - t_1} \mid n \in \nodesg, t_1 \in \tdg \te{and} t_2 \in (\delta + t_1) \cap \tdg\}
    \end{align*}
    \begin{itemize}
    \item For~\ref{item:tcorr_1}, let $\vt = \tup{n,n, t, t+d} \in \eval{q}$.\\
      And let $\u = \tup{n,n, [t,t], d}$ in $\tuplest$.\\
      For \ref{item:tcorr_1_1}
      we show that $\u \in \evalcit{q}$.\\
      From $\vt \in \eval{q}$,
      we get $n \in \nodesg$ and $t \in \tdg$.\\
      Besides,
      because $\vt \in \eval{q}$ still,
      \begin{equation}
        \label{eq:t0_0}
        t+d \in \tdg
      \end{equation}
      and
      \begin{align}
        d &\in \delta\\
        t + d &\in t + \delta \label{eq:t0_1}
      \end{align}
      So from~\eqref{eq:t0_0} and~\eqref{eq:t0_1} 
      \begin{equation}
      t+d \in (\delta + t) \cap \tdg
      \end{equation}
      So there is a $t_2$ (namely $t + d$) such that $d = t_2 - t$ and $t_2 \in t + \delta \cap \tdg$.\\
      Together with the definition of $\evalcit{q}$,
      this implies $\u \in \evalcit{q}$,
      which concludes the proof for \ref{item:tcorr_1_1}.\\
    And trivially, $t \in [t,t]$,
    so~\ref{item:tcorr_1_2} is verified as well.
    
  \item For~\ref{item:tcorr_2}, let
    $\u = \tup{n,n,[t, t], d} \in \evalcit{q}$.\\
    From $\u \in \evalcit{q}{G}$,
    we get $n \in \nodesg$ and $t \in \tdg$.\\
    So to conclude the proof,
    it is sufficient to show that \ei $d \in \delta$ and \eii $t + d \in \tdg$.\\
    Because $\u \in \evalcit{q}{G}$ still, we have
    \begin{equation} \label{eq:tt_1_0}
      d = t_2 - t \te{for some} t_2 \in (\delta + t) \cap \tdg
    \end{equation} 
    From~\eqref{eq:tt_1_0}, we get
    $t_2 = t + d$.\\
    Therefore from~\eqref{eq:tt_1_0} still,
    \begin{equation} \label{eq:tt_1_1}
    t + d \in (\delta + t) \cap \tdg
    \end{equation}
    which proves \eii.\\
    And from~\eqref{eq:tt_1_1}, we also get
    \begin{align*}
    t + d&\in \delta + t\\
    t + d - t &\in (\delta + t) - t\\
    d &\in \delta
    \end{align*}
    which proves \ei.\\
  \end{itemize}


  \item $q = (?\query)$.\ \\
    From the above definitions, we have:
    \begin{align*}
      \eval{q} = &\ \{\tup{n,n,t,0} \mid \tup{n,n',t,t+d} \in \eval{\query} \te{for some} n' \in \nodesg, d \in \td\}\\
  \evalcit{q} = &\ \{\tup{n,n, \tau, 0} \mid  \tup{n, n', \tau, d} \in \evalcit{\query} \te{for some} n' \in \nodesg, d \in \td \}\\
    \end{align*}

    \begin{itemize}
    \item For~\ref{item:tcorr_1}, let $\tup{n,n,t,0} \in \eval{q}$.\\
      From the definition of $\eval{q}$,
      there are $n'$ and $d$ such that $\tup{n,n',t, t+d } \in \eval{\query}$.\\
      So by IH, there is a $\tau$ s.t.~$t \in \tau$ and $\tup{n, n', \tau, d} \in \evalcit{\query}$.\\
      Therefore $\tup{n,n,\tau,0} \in \evalcit{q}$,
      from the definition of $\evalcit{q}$.
     \item For~\ref{item:tcorr_2}, let $\tup{n, n, \tau, 0} \in \evalcit{q}$.\\
       From the definition of $\evalcit{q}$, there are $n'$ and $d$ s.t.~$\tup{n, n', \tau, d} \in \evalcit{\query}$.\\
       Now take any $t\in \tau$.\\
       By IH, $\tup{n,n',t,t+d} \in \eval{\query}$.\\
       Therefore $\tup{n,n,t,0} \in \eval{q}$,
      from the definition of $\eval{q}$.\\
    \end{itemize} 

\item $q = \neg \node$.\ \\
    From the above definitions, we have:
    \begin{align*}
  \eval{q}\ = & (\{\tup{n,n} \mid n \in \nodesg\} \times \tdg \times \{0\}) \setminus \eval{\node}\\
  \evalcit{q} =& \bigcup\limits_{n \in \nodesg}\Big\{\tup{n,n,\tau, 0} \mid \tau \in \compl{\{\tau' \mid \tup{n,n,\tau', 0} \in \evalcit{\node}\}}{\tdg}\Big\}\\
    \end{align*}
    \begin{itemize}
    \item For~\ref{item:tcorr_1}, let $\vt = \tup{n, n,t,0} \in \eval{q}$.\\
      From the definition of $\eval{q}$,
      $\vt \not\in \eval{\node}$.\\
      So
      \begin{equation}\label{eq:t4_1}
      t \notin \{t' \mid \tup{n,n,t',0} \in \eval{\node}\}
      \end{equation}
      Now by IH, together with Lemma~\ref{lemma:test_dist_zero}, we get:
      \begin{equation}\label{eq:t4_2}
        \tup{n,n,t',0} \in \eval{\node} \te{iff} t' \in \tau' \te{for some} \tau' \te{s.t.} \tup{n,n, \tau',0} \in \evalcit{\node} 
      \end{equation}
      So from~\eqref{eq:t4_1} and~\eqref{eq:t4_2}:
      \begin{equation*}
        t \notin \bigcup \{\tau' \mid \tup{n,n, \tau',0} \in \evalcit{\node}\}
      \end{equation*}
      Therefore
      \begin{equation}\label{eq:t4_2_1}
        t \in \tdg \setminus \bigcup \{\tau' \mid \tup{n,n, \tau',0} \in \evalcit{\node}\}
      \end{equation}
      So $t \in \tau$ for some $\tau \in \compl{\bigcup \{\tau' \mid \tup{n,n, \tau',0} \in \evalcit{\node}\}}{\tdg }$.\\
      And
      $\tup{n,n, \tau, 0} \in \evalcit{q}$,
      from the definition of $\evalcit{q}$.
    
  \item For~\ref{item:tcorr_2}, let $\tup{n,n, \tau,0} \in \evalcit{q}$.\\
      And take any $t \in \tau$.\\
      From the definition of $\evalcit{q}$:
      \begin{equation*}
      t \in \tdg \setminus \bigcup \{\tau' \mid \tup{n,n, \tau',0} \in \evalcit{\node} \}
      \end{equation*}
      Together with~\eqref{eq:t4_2}, this implies
      \begin{equation*}
        \tup{n,n,t,0} \not \in \eval{\node}
      \end{equation*}
      Therefore
      $\tup{n,n, t, 0} \in \eval{q}$,
      from the definition of $\eval{q}$.\\
  \end{itemize}

\item $q = \query_1 / \query_2$.\ \\
    From the above definitions, we have:
\[
\begin{array}{rl}
  \eval{q}\ = &\{\tup{n_1, n_3 , t, d_1 + d_2 } \mid \exists n_2 \colon \tup{n_1, n_2 , t, d_1} \in  \eval{\query_1} \te{and} \tup{n_2, n_3, t + d_1, d_2} \in  \eval{\query_2}\}\\
  \evalcit{q} = & \Big\{\tup{n_1, n_3, ((\tau_1 + d_1) \cap \tau_2)- d_1, d_1 + d_2} \mid \\
  \multicolumn{2}{l}{\qquad \qquad \qquad
    \exists n_2 \colon \tup{n_1, n_2, \tau_1, d_1} \in \evalcit{\query_1}, \tup{n_2, n_3, \tau_2, d_2} \in  \evalcit{\query_2} \te{and} 
  (\tau_1 + d_1) \cap \tau_2 \neq \emptyset\ \Big\}}\\
\end{array}
\]
\begin{itemize}
    \item For~\ref{item:tcorr_1}, let $\vt = \tup{n_1,n_3, t, d} \in \eval{q}$.\\
      Fom the definition of $\eval{q}$, there are  $n_2, d_1$ and $d_2$ such that
      $\tup{n_1, n_2 , t, d_1} \in  \eval{\query_1}$,
      $\tup{n_2, n_3 , t + d_1, d_2} \in  \eval{\query_2}$
      and $d = d_1 + d_2$.\\
      By IH, because $\tup{n_1, n_2, t , d_1} \in  \eval{\query_1}$,
      there is a $\tau_1$ such that $t \in \tau_1$ and
      \begin{equation}\label{eq:t2_12}
      \tup{n_1, n_2, \tau_1, d_1} \in \evalcit{\query_1}\\
    \end{equation}
     And similarly,
      because $\tup{n_2, n_3 , t + d_1, d_2} \in  \eval{\query_2}$,
      there is a $\tau_2$ such that $t + d_1 \in \tau_2$ and
      \begin{equation}\label{eq:t2_2}
        \tup{n_2, n_3, \tau_2, d_2} \in \evalcit{\query_2}
    \end{equation} 
    From $t \in \tau_1$, we get
      \begin{equation}\label{eq:t2_21}
        t + d_1 \in \tau_1 + d_1
      \end{equation}
      Together with the fact that $t + d_1 \in \tau_2$, this implies
      \begin{equation}\label{eq:t2_3}
       \tau_1 + d_1 \cap \tau_2 \neq \emptyset
      \end{equation}
      So from~\eqref{eq:t2_12},~\eqref{eq:t2_2},~\eqref{eq:t2_3} and the definition of $\evalcit{q}$,
      \begin{equation*}
      \tup{n_1, n_2, ((\tau_1 + d_1) \cap \tau_2) -d_1 , d_1 + d_2} \in \evalcit{q}
    \end{equation*}
    which proves~\ref{item:tcorr_1_1}.\\
    And in order to prove~\ref{item:tcorr_1_2}, we only need to show that
    \[t \in ((\tau_1 + d_1) \cap \tau_2) -d_1\]
    We know that $t \in \tau_1$, therefore
    \begin{equation*}
      t + d_1 \in \tau_1 + d_1
    \end{equation*}
    Together with the fact that $t + d_1 \in \tau_2$, this yields
    \begin{align*}
      t + d_1 &\in (\tau_1 + d_1) \cap \tau_2\\
      t &\in ((\tau_1  + d_1) \cap \tau_2) -d_1
    \end{align*}
  \item For~\ref{item:tcorr_2}, let $\u = \tup{n_1,n_3,\tau,d} \in \evalcit{q}$, and let $t \in \tau$.\\
    We show that $\tup{n_1, n_3, t, t + d} \in \eval{q}$.\\
    Because $\u \in \evalcit{q}$, from the definition of $\evalcit{q}$, there are $\tau_1, \tau_2, d_1, d_2$ and $n_2$ s.t.:
    \begin{enumerate}[(i)]
    \item $d = d_1 + d_2$\label{item:t3_0_0}
    \item $\tau = ((\tau_1 + d_1) \cap \tau_2) - d_1$\label{item:t3_0_1}
    \item $\tup{n_1, n_2, \tau_1, d_1} \in  \evalcit{\query_1}$\label{item:t3_1}
    \item $\tup{n_2, n_3, \tau_2, d_2} \in  \evalcit{\query_2}$\label{item:t3_2}
    \end{enumerate} 
    Since $t \in \tau$, from~\ref{item:t3_0_1}, we have
    \begin{align}
      t &\in ((\tau_1 + d_1 \cap \tau_2) - d_1\\
      t + d_1 &\in (((\tau_1 + d_1 \cap \tau_2) - d_1) + d_1\\
      t + d_1 &\in (\tau_1 + d_1) \cap \tau_2 \label{eq:t3_01}\\
      t + d_1 &\in \tau_1 + d_1 \\
      t  &\in \tau_1\label{eq:t3_1}
    \end{align}
    From~\ref{item:t3_1}, by IH, for any $t' \in \tau_1$
    \begin{equation*}
    \tup{n_1, n_2, t' + d_1} \in \eval{q}
  \end{equation*}
  In particular, from~\eqref{eq:t3_1}
    \begin{equation}\label{eq:t3_2}
    \tup{n_1, n_2, t, t + d_1} \in \eval{q}
  \end{equation}
    And from~\ref{item:t3_2}, by IH, for any $t'' \in \tau_2$
    \begin{equation*}
    \tup{n_2, n_3, t'', t'' + d_2} \in \eval{q}
  \end{equation*}
  In particular, from~\eqref{eq:t3_01}
    \begin{equation}\label{eq:t3_3}
    \tup{n_2, n_3, t + d_1, (t + d_1) + d_2} \in \eval{q}
  \end{equation}
  So from~\eqref{eq:t3_2},~\eqref{eq:t3_3} and the definition of $\eval{q}$
  \begin{equation*}
    \tup{n_1, n_3, t, t + d_1 + d_2} \in \eval{q}
  \end{equation*}
\end{itemize}
    \end{itemize} 
\end{proof}

\subsubsection{In $\tuplesd$}
\label{sec:evald}
\paragraph{Definition.}
\label{sec:evald_def}
We start with the case where $q$
is an expression for the symbol $\node$ or $\edge$ in the grammar of Section~\ref{sec:language}.
\renewcommand{\arraystretch}{1.4}
\[
  \begin{array}{rl}
    \evalcid{\edgelabel} =& \{\tup{n_1, n_2, t, [0,0]} \mid t \in \tau \te{for some} \tau \in \val{n_1}{\edgelabel}{n_2} \}\\
    \evalcid{\edge^-} =& \{\tup{n_2, n_1, t, [0,0]} \mid \tup{n_1, n_2, t, [0,0]} \in \evalcid{\edge} \}\\
    \evalcid{\pred} =& \{\tup{n, n, t, [0,0]} \mid n \models \pred \te{and} t \in \tdg \}\\
    \evalcid{\le k} =& \{\tup{n, n, t , [0,0]} \mid n \in \nodesg, t \in \tdg \te{and} t \le k \}\\
%
%
    \evalcid{(?\query)} = & \{\tup{n,n, t, [0,0]} \mid \delta \colon \tup{n, n', t, \delta} \in \evalcid{\query} \te{for some} n' \in \N \te{and} \delta \in \intervals{\td}\}\\
    \evalcid{\neg \node} = &\Big\{\tup{n,n,t, [0,0]} \mid n \in \nodesg \te{and} t \in \tdg \setminus \{t' \mid \tup{n,n,t', [0,0]} \in \evalcid{\node}\}\Big\}
\end{array} 
\]
\renewcommand{\arraystretch}{1}
Next, we consider the operators
($\query_1 + \query_2$), ($\query[m,\_]$) and ($\query[m,n]$).\\
For these cases,
$\evalcitd{q}$
is once again defined analogously to $\eval{q}$,
in terms of temporal join (a.k.a. $\query_1/\query_2$) and set union.\\
We only write the definitions here for the sake of completeness:
\[
  \begin{array}{rl}
  \evalcitd{\query_1 + \query_2} = & \evalcitd{\query_1} \cup \evalcitd{\query_2}\\
  \eval{\query[m,n]} = & \bigcup\limits_{k = m}^n\evalcitd{\query^k}\\
  \eval{\query[m,\_]} = &  \bigcup\limits_{k \ge m}\evalcitd{\query^k}\\
\end{array}
\]
The only remaining operators are temporal join ($\query_1/\query_2$) and temporal navigation ($\tndelta$), already defined in the article.\\
We reproduce here these two definition for convenience:
\[
  \begin{array}{rl}
    \evalcid{\query_1/\query_2} = & \Big\{\tup{n_1, n_3, t_1, t_2 - t_1 + \delta_2} \mid \exists n_2 \colon \tup{n_1, n_2, t_1 , \delta_1} \in \evalcid{\query_1},\\ & \qquad \tup{n_2, n_3, t_2, \delta_2} \in  \evalcid{\query_2} \te{and} t_2 - t_1 \in \delta_1 \Big\}\\
\evalcid{\tndelta} = &\Big\{\tup{n,n,t, ((\delta + t) \cap \tdg) - t } \mid n \in \nodesg, t \in \tdg \te{and} (\delta + t) \cap \tdg \neq \emptyset\Big\}\\
  \end{array}
\]
We also reproduce the alternative characterization of $\evalcid{\tndelta}$ provided in the article, as a unary operator: 
\[\evalcid{q/\tndelta} =  \Big\{\tup{n_1,n_2,t, (\delta' \oplus \delta) \cap \tdg} \mid \tup{n_1,n_2,t, \delta'} \in \evalcid{q} \te{and} (t + (\delta' \oplus \delta)) \cap \tdg \neq \emptyset\Big\}\]

\paragraph{Correctness.}
The following result states that the representation $\evalcid{q}$ is correct:
\begin{restatable}{proposition}{correctD}\label{prop:correct_d}
  Let $G = \tup{\tdg, \fg, \valz}$ be a temporal graph and $q$ a TRPQ. Then the unfolding of $\evalcid{q}$ is $\eval{q}$.
\end{restatable}

\begin{proof}\ \\
  Let $G = \tup{\tdg \fg, \valz}$ be a temporal graph,
  and let $q$ be a TRPQ.\\
  We show below that:
  \begin{enumerate}[(I)]
  \item for any $\tup{n_1,n_2,t, d} \in \eval{q}$,
    there is a $\delta \in \intervals{\td}$ such that\label{item:dcorr_1}
    \begin{enumerate}[(a)]
    \item $\tup{n_1,n_2, t, \delta} \in \evalcid{q}$\label{item:dcorr_1_1}, and
    \item $d \in \delta$,\label{item:dcorr_1_2}
    \end{enumerate}
  \item for any $\tup{n_1,n_2,t, \delta} \in \evalcid{q}$
    for any $d \in \delta$,\\
    \qquad $\tup{n_1,n_2, t, d}$ is in $\eval{q}$.
    \label{item:dcorr_2}
  \end{enumerate}
  
    \nt We proceed once again by induction on the structure of $q$.\\
    If $q$ is of the form $\edgelabel$, $\edge^-$, $\pred$, $\le k$, $\neg \node$, $(\query + \query)$, $\query[m,n]$ or $\query[m,\_]$,
    then~\ref{item:dcorr_1} and~\ref{item:dcorr_2} immediately follow from the definitions of
    $\eval{q}$ and 
    $\evalcid{q}$.\\
    If $q$ is of the form $(?\query)$, then the proof is nearly identical to the one already provided for $\evalcit{(?\query)}$.\\
    So we focus below on the two remaining cases:
  \begin{itemize}
  \item $q = \tndelta$.\ \\
    From the definitions above, we have:
    \begin{align*}
      \eval{q}\ = &\{\tup{n, n, t, d} \mid n \in \nodesg, t \in \tdg, d \in \delta \te{and} t + d \in \tdg\}\\
      \evalcid{q} = &\{\tup{n,n,t, ((\delta + t) \cap \tdg) -t } \mid n \in \nodesg, t \in \tdg \te{and} (\delta + t) \cap \tdg \neq \emptyset\}
    \end{align*}
    \begin{itemize}
    \item For~\ref{item:dcorr_1}, let $\vt = \tup{n,n, t, d} \in \eval{q}$.\\
      And let $\u = \tup{n,n, t, ((\delta + t) \cap \tdg) -t}$ in $\tuplesd$.\\
      For \ref{item:dcorr_1_1}
      we show that $\u \in \evalcid{q}$.\\
      From $\vt \in \eval{q}$,
      we get $n \in \nodesg$ and $t \in \tdg$.\\
      Besides,
      because $\vt \in \eval{q}$ still,
      \begin{equation}
        \label{eq:d0_0}
        t+d \in \tdg
      \end{equation}
      and
      \begin{align}
        d &\in \delta\\
        t + d &\in t + \delta \label{eq:d0_1}
      \end{align}
      So from~\eqref{eq:d0_0} and~\eqref{eq:d0_1} 
      \begin{align}
      t+d &\in  (\delta + t) \cap \tdg~\label{eq:d0_2}\\
      (\delta + t) \cap \tdg &\neq  \emptyset
      \end{align}
      Together with the definition of $\evalcid{q}$,
      this implies $\u \in \evalcid{q}$,
      which concludes the proof for \ref{item:dcorr_1_1}.\\
      Finally, from~\eqref{eq:d0_2}, we get
      \begin{align}
        t+d -t &\in  ((\delta + t) \cap \tdg) - t\\
        d&\in  ((\delta + t) \cap \tdg) - t
      \end{align}
      which proves~\ref{item:dcorr_1_2}.
    
  \item For~\ref{item:dcorr_2}, let
    $\u = \tup{n,n,t, ((\delta + t) \cap \tdg) -t} \in \evalcid{q}$, and let $d \in  ((\delta + t) \cap \tdg) -t$.\\
    From $\u \in \evalcid{q}{G}$,
    we get $n \in \nodesg$ and $t \in \tdg$.\\
    So to conclude the proof,
    it is sufficient to show that \ei $d \in \delta$ and \eii $t + d \in \tdg$.\\
    By assumption, we have
    \begin{align} 
      d &\in ((\delta + t) \cap \tdg) - t\\
      d + t &\in (\delta + t) \cap \tdg\label{eq:dt_1_0}\\
      d + t &\in \tdg
    \end{align}
    which proves \eii.\\
    And from~\eqref{eq:dt_1_0}, we also get
    \begin{align*}
    d + t&\in \delta + t\\
    d + t - t &\in (\delta + t) - t\\
    d &\in \delta
    \end{align*}
    which proves \ei.\\
  \end{itemize}

\item $q = \query_1 / \query_2$.\ \\
    From the definitions above, we have:
\[
  \begin{array}{rl}
    \eval{q}\ = &\{\tup{n_1, n_3 , t, d_1 + d_2 } \mid \exists n_2 \colon \tup{n_1, n_2 , t, d_1} \in  \eval{\query_1} \te{and} \tup{n_2, n_3, t + d_1, d_2} \in  \eval{\query_2}\}\\
\evalcid{q} = &
\Big\{\tup{n_1, n_3, t_1, \delta_2 + t_2 - t_1} \mid \exists n_2 \colon\\
              & \tup{n_1, n_2, t_1 , \delta_1} \in \evalcid{\query_1}, \tup{n_2, n_3, t_2, \delta_2} \in  \evalcid{\query_2} \te{and}
  t_2 - t_1 \in \delta_1 \Big\}
\end{array}
\]
\begin{itemize}
    \item For~\ref{item:dcorr_1}, let $\vt = \tup{n_1,n_3, t, d} \in \eval{q}$.\\
      Fom the definition of $\eval{q}$, there are  $n_2, d_1$ and $d_2$ such that
      $\tup{n_1, n_2 , t, d_1} \in  \eval{\query_1}$,
      $\tup{n_2, n_3 , t + d_1, d_2} \in  \eval{\query_2}$
      and $d = d_1 + d_2$.\\
      By IH, because $\tup{n_1, n_2, t , d_1} \in  \eval{\query_1}$,
      there is a $\delta_1$ such that $d_1 \in \delta_1$ and
      \begin{equation}\label{eq:d2_12}
      \tup{n_1, n_2, t, \delta_1} \in \evalcid{\query_1}
    \end{equation}
    
    And similarly,
      because $\tup{n_2, n_3 , t + d_1, d_2} \in  \eval{\query_2}$,
      there is a $\delta_2$ such that $d_2 \in \delta_2$ and
      \begin{equation}\label{eq:d2_2}
        \tup{n_2, n_3, t + d_1, \delta_2} \in \evalcid{\query_2}
      \end{equation}
      Next, since $d \in \delta_1$
      \begin{equation}\label{eq:d2_2_1}
        t + d_1 \in t + \delta_1
      \end{equation}
      So from~\eqref{eq:d2_12},~\eqref{eq:d2_2},~\eqref{eq:d2_2_1} and the definition of $\evalcid{q}$ (replacing $t_1$ with $t$ and $t_2$ with $t + d_1$),
      we get
      \begin{equation*}\label{eq:d2_22}
      \tup{n_1, n_2, t, \delta_2 + (t + d_1) - t} \in \evalcid{q}
    \end{equation*}
    which proves~\ref{item:dcorr_1_1}.\\
    And in order to prove~\ref{item:dcorr_1_2}, we only need to show that $d \in \delta_2 + (t + d_1) - t$,
    or in other words that
    \[d \in \delta_2 + d_1\]
    We know that
    \begin{align}
      d_2  &\in \delta_2\\
      d_2 + d_1  &\in \delta_2 + d_1\label{eq:d2_23}
    \end{align}
    Together with the fact that $d = d_1 + d_2$,
    this concludes the proof for ~\ref{item:dcorr_1_2}.

  \item For~\ref{item:dcorr_2}, let $\u = \tup{n_1,n_3,t_1,\delta} \in \evalcid{q}$, and let $d \in \delta$.\\
    
    Because $\u \in \evalcid{q}$, from the definition of $\evalcid{q}$, there are $\delta_1, \delta_2, t_2$ and $n_2$ s.t.:
    \begin{enumerate}[(i)]
    \item $\delta = \delta_2 + t_2 - t_1$\label{item:d3_0_0}
    \item $t_2 \in  t_1 + \delta_1$\label{item:d3_0_1}
    \item $\tup{n_1, n_2, t_1, \delta_1} \in  \evalcid{\query_1}$\label{item:d3_1}
    \item $\tup{n_2, n_3, t_2, \delta_2} \in  \evalcid{\query_2}$\label{item:d3_2}
    \end{enumerate} 
    From~\ref{item:d3_0_0} and~\ref{item:d3_0_1}, we get
    \begin{align*}
      \delta &= \delta_2 \oplus (t_1 + \delta_1) - t_1\\
             &= \delta_2 \oplus \delta_1
    \end{align*}
    Together with $d \in \delta$, this implies that there are 
$d_1 \in \delta_1$ and 
$d_2 \in \delta_2$ such that $d = d_1 + d_2$. \\
Next, because $d_1 \in \delta_1$, from~\ref{item:d3_1}, by IH
\begin{equation}\label{eq:d3_2}
    \tup{n_1, n_2, t_1, t_1 + d_1} \in \eval{q}
  \end{equation}
And similarly, because $d_2 \in \delta_2$, from~\ref{item:d3_2}
\begin{equation}\label{eq:d3_3}
    \tup{n_2, n_3, t_2, t_2 + d_2} \in \eval{q}
  \end{equation}
  So from~\eqref{eq:d3_2},~\eqref{eq:d3_3} and the definition of $\eval{q}$ 
\begin{equation}
    \tup{n_1, n_3, t_1, d_1 + d_2} \in \eval{q}
  \end{equation}
Together with the fact that $d = d_1 + d_2$, this concludes the proof for~\ref{item:dcorr_2}.
\end{itemize}

    \end{itemize} 
\end{proof}

\subsubsection{In $\tuplestd$}
\label{sec:evaltd}
\paragraph{Definition.}
We start with the case where $q$
is an expression for the symbol $\edge$ or $\node$ in the grammar of Section~\ref{sec:language}.\\
As a consequence of Lemma~\ref{lemma:test_dist_zero},
$\evalcitd{q}$ can be trivially defined out of $\evalcit{q}$ by replacing the distance $0$ with the interval $[0,0]$, i.e.
\begin{align*}
  \evalcitd{\edge} = & \{\tup{n_1, n_2, \tau, [0,0]} \mid  \{\tup{n_1, n_1, \tau, 0} \in \evalcit{\edge}\}\\
  \evalcitd{\node} = & \{\tup{n, n, \tau, [0,0]} \mid  \{\tup{n, n, \tau, 0} \in \evalcit{\node}\}
\end{align*}
Next, if $q$ is of the form
($\query_1 + \query_2$), ($\query[m,\_]$) or ($\query[m,n]$),
then the definition of $\evalcitd{q}$ is once again nearly identical to the one of $\eval{q}$:
\[
  \begin{array}{rl}
  \evalcitd{\query_1 + \query_2} = & \evalcitd{\query_1} \cup \evalcitd{\query_2}\\
  \eval{\query[m,n]} = & \bigcup\limits_{k = m}^n\evalcitd{\query^k}\\
  \eval{\query[m,\_]} = &  \bigcup\limits_{k \ge m}\evalcitd{\query^k}\\
\end{array}
\]
The only remaining operators are temporal join ($\query_1/\query_2$) and temporal navigation ($\tndelta$):
\[
  \begin{array}{rl}
\evalcitd{\query_1/\query_2} =& \bigcup \{\u_1 \tjoin \u_2 \mid \u_1 \in \evalcitd{\query_1}, \u_2 \in \evalcitd{\query_2}\}\\
\evalcitd{\tndelta} =& \bigcup\limits_{n \in \nodesg} \{\tup{n,n,\tdg, \delta} \tjoin \tup{n,n,\tdg, [0,0]}\}
  \end{array}
  \]
where $\u_1 \tjoin \u_2$ is defined as follows.\\
Let $\u_1 = \tup{n_1, n_2, \tau_1, \delta_1}$ and $\u_2 = \tup{n_3, n_4, \tau_2, \delta_2}$.\\
Define $\tau'_2$ as
\[ \tau_2' = (\tau_1 \oplus \delta_1) \cap \tau_2\]
If $n_2 \neq n_3$ or $\tau_2' = \emptyset$,
then $\u_1 \tjoin \u_2 = \emptyset$.\\
Otherwise, let:
\begin{align*}
\tau & = (\tau_2' \ominus \delta_1) \cap \tau_1\\
b & = b_{\tau'_2} - b_{\delta_1}\\
e & = e_{\tau'_2} - e_{\delta_1}
\end{align*}
And for every $t\in \tau$, let
\[\delta(t) =\ \  \ld{\delta_1}\ b_{\delta_1} + \max(0, b-t),\ e_{\delta_1} - \max(0, t-e)\ \rd{\delta_1} \]
Then
\[\u_1 \tjoin \u_2 = \{\tup{n_1, n_4, [t,t], \delta(t) \oplus \delta_2} \mid t \in \tau\}\]







\paragraph{Correctness.}
\label{sec:evaltd_correct}
We start with a lemma:
\begin{lemma}\label{lemma:ominus}
  Let $\alpha, \beta \in \intervals{\td}$.
  Then
  \[\beta \ominus \alpha = \{t \mid (t + \alpha) \cap \beta \neq \emptyset\}\]
\end{lemma}

\nt Next, if $\u = \tup{n_1, n_2, \tau, \delta} \in \tuplestd$, we call \emph{temporal relation induced by} $\u$ the
set
\[\{(t, t+d) \mid t \in \tau, d \in \delta\}\]

\nt We also use the operator $\join\colon (\td \times \td) \times (\td \times \td) \to (\td \times \td) $ for the composition of two binary relations, i.e.
\[R_1 \join R_2 = \{t_1, t_3 \mid (t_1, t_2) \in R_1 \te{and} (t_2, t_3) \in R_2 \te{for some} t_2\}\]
We can now formulate the following lemma:
\begin{lemma}\label{lemma:tjoin_td}
  Let $\u_1 = \tup{n_1, n_2, \tau_1, \delta_1}$
  and $\u_2 = \tup{n_3, n_4, \tau_2, \delta_2}$ be two tuples in $\tuplestd$ such that $n_2 = n_3$.
  And for $i \in \{1,2\}$, let $R_i$ denote the temporal relation induced by $\u_i$.
  Then
  \[R_1 \join R_2 = \bigcup_{\tup{n_1, n_4, \tau, \delta}\in \u_1 \tjoin \u_2} \{(t, t+d) \mid t \in \tau, d \in \delta\}\]
\end{lemma}
\begin{proof}
  $\u_1 = \tup{n_1, n_2, \tau_1, \delta_1}$
  and $\u_2 = \tup{n_3, n_4, \tau_2, \delta_2}$ be two tuples in $\tuplestd$ such that $n_2 = n_3$.
  And for $i \in \{1,2\}$, let $R_i$ denote the temporal relation induced by $\u_i$.
  
  We show that:
  \begin{enumerate}[(I)]
  \item\label{item:tdj0_1}
    \begin{enumerate}[(a)]
    \item If $\tau_2' = \emptyset$, then $\dom{R_1 \join R_2} = \emptyset$,\label{item:tdj0_1_1}
      \item 
    otherwise $\tau = \dom{R_1 \join R_2}$,\label{item:tdj0_1_2}
    \end{enumerate}
  \item for each $t \in \tau$,\label{item:tdj0_2}
    \[t + \delta(t) \oplus \delta_2 = \{t' \mid (t,t') \in R_1 \join R_2\}\]
  \end{enumerate}
  We start with~\ref{item:tdj0_1}.\\
  From the definition of $\oplus$:
  \begin{equation}
    \tau_1 \oplus \delta_1  = \{t +d \mid t \in \tau_1, d_1 \in \delta_1\}
  \end{equation}
  So from the definition of $R_1$
  \begin{equation}
    \tau_1 \oplus \delta_1  = \range{R_1}
  \end{equation}
  Since $\tau_2 = \dom{R_2}$, this implies
  \begin{align}
    (\tau_1 \oplus \delta_1) \cap \tau_2 = \range{R_1} \cap \dom{R_2}\\
    \tau_2' = \range{R_1} \cap \dom{R_2}\label{eq:tdj_1}
  \end{align}
  If $\range{R_1} \cap \dom{R_2} = \emptyset$, then $\dom{R_1 \join R_2} = \emptyset$, immediately  from the definition of $\join$,
  which concludes the proof of~\ref{item:tdj0_1_1}. \\
  Otherwise, from Lemma~\ref{lemma:ominus},
  \begin{equation}
  \tau_2' \ominus \delta_1  = \{t \mid (t + \delta_1) \cap \tau_2' \neq \emptyset\}\\
  \end{equation}
  So from~\eqref{eq:tdj_1}
  \begin{align*}
  \tau_2' \ominus \delta_1 & = \{t \mid (t + \delta_1) \cap \range{R_1} \cap \dom{R_2} \neq \emptyset\}\\
     (\tau_2' \ominus \delta_1) \cap \tau_1 & = \{t \in \tau_1 \mid (t + \delta_1) \cap \range{R_1} \cap \dom{R_2} \neq \emptyset\}\\
    (\tau_2' \ominus \delta_1) \cap \tau_1 & = \dom{R_1 \join R_2}\\
    \tau & = \dom{R_1 \join R_2}
  \end{align*}
  which proves~\ref{item:tdj0_1_2}.\\

  \nt Now for~\ref{item:tdj0_2},
  let $t \in \tau$.\\\
  We show below that \ei $t + \delta(t) = \{t' \mid (t,t') \in R_1 \te{and} t' \in \range{R_1} \cap \dom{R_2}\}$.\\
  Together with the definition of $\join$ (and the fact that $t + \delta(t)$ is an interval), this proves~\ref{item:tdj0_2}.\\

  \nt We only prove the result for the case where $\tau$, $\tau_2'$ and $\delta_1$ are closed-closed intervals (the proof for the other 63 cases is symmetric).\\
  First, from~\ref{item:tdj0_1_2} and the assumption that $t \in \tau$,
  we have $t \in \tau_1$.
  So from the definition of $R_1$,
  \begin{equation}
    t + \delta_1 = \{t' \mid (t,t') \in R_1\}
  \end{equation}
  Together with~\eqref{eq:tdj_1}, this means that \ei is equivalent to
  \eii $t + \delta(t) = \{(t + \delta_1) \cap \tau_2' \}$.\\
  So in order to prove~\ref{item:tdj0_2} (and conclude our proof), it is sufficient to prove \eii.\\
  
\nt Now since $t \in \tau$, from~\ref{item:tdj0_1_2} and the definition of $\tau_2'$, we have 
$(t + \delta(t)) \cap \tau_2' \neq \emptyset$.\\
And since $\delta(t)$ and $\tau_2'$ are intervals, $(t + \delta(t)) \cap \tau_2'$ is an also an interval.\\
So in order to prove \eii, it is sufficient to show that $t + b_{\delta(t)}$ (resp.~$t + e_{\delta(t)}$) is the smallest (resp. greatest) value in
$(t + \delta_1) \cap \tau_2'$.\\
We only prove the result for $t + b_{\delta(t)}$ (the proof for $t + e_{\delta(t)}$) is symmetric.\\
We consider two cases.
\begin{itemize}
\item If $b \le t$, then
 \begin{align}
   b_{\tau'_2} - b_{\delta_1} &\le t& \te{from the definition of} b\\
   b_{\tau'_2} - b_{\delta_1} + b_{\delta_1} &\le t + b_{\delta_1} & \\
   b_{\tau'_2} &\le t + b_{\delta_1}\label{eq:tdj2_1_1}
 \end{align}
 And because $t \in \tau$
 \begin{align}
   t &\le e_\tau\\
   t &\le e_{\tau_2'} - b_{\delta_1} & \te{from the definition of} \tau\\
   t + b_{\delta_1} &\le e_{\tau_2'} - b_{\delta_1} + b_{\delta_1}\\
   t + b_{\delta_1} &\le e_{\tau_2'}\label{eq:tdj2_1_2}
 \end{align} 
 So from~\eqref{eq:tdj2_1_1} and~\eqref{eq:tdj2_1_2} 
 \begin{equation}
   t + b_{\delta_1} \in \tau_2'\label{eq:tdj2_2}
 \end{equation}
 Next, since  $b \le t$ (by assumption), we have
 \begin{align*}
   b - t \le 0\\
  \max(0, b - t) = 0
 \end{align*}
  So from the definition of $\delta(t)$
  \begin{equation}
 b_{\delta(t)} = b_{\delta_1}~\label{eq:tdj2_3}
  \end{equation}
  Therefore $t + b_{\delta(t)}$ is the smallest value in $t + \delta_1$.\\
  So from~\eqref{eq:tdj2_2},
  it is also the smallest value in $t + \delta_1 \cap \tau_2'$, which concludes the proof for this case.

\item If $b > t$, then
 \begin{align}
   b - t > 0\label{eq:tdj2_4}\\
   \max(0, b - t) =  b - t
 \end{align}
 So from the definition of $\delta(t)$
 \begin{equation}
    b_{\delta(t)} = b_{\delta_1} + b - t\label{eq:tdj2_4_1}
  \end{equation}
  Besides, from~\eqref{eq:tdj2_4}
  \begin{equation}
  b - t + b_{\delta_1} > b_{\delta_1}\label{eq:tdj2_4_2}
  \end{equation}
  So from~\eqref{eq:tdj2_4_1} and~\eqref{eq:tdj2_4_2}
  \begin{equation}
  b_{\delta(t)} > b_{\delta_1}\label{eq:tdj2_5}
  \end{equation}
  Next, since $t \in \tau$
  \begin{equation}
    b_\tau \le t\label{eq:tdj2_6_1}
  \end{equation}
  And from the definition of $\tau$
  \begin{equation}
    b_{\tau_2'} - e_{\delta_1} \le  b_{\tau}\label{eq:tdj2_6_2}
  \end{equation}
  So from~\eqref{eq:tdj2_6_1} and~\eqref{eq:tdj2_6_2}
  \begin{align}
    b_{\tau_2'} - e_{\delta_1} &\le  t\\
    b_{\tau_2'} - t  &\le  e_{\delta_1}\\
    b_{\tau_2'} - t  + b_{\delta_1} -  b_{\delta_1} &\le  e_{\delta_1}\\
    b_{\delta_1}  + (b_{\tau_2'} -  b_{\delta_1}) - t   &\le  e_{\delta_1}\\
    b_{\delta_1}  + b - t   &\le  e_{\delta_1}& \te{from the definition of} b\\
    b_{\delta(t)} &\le  e_{\delta_1}& \te{from~\eqref{eq:tdj2_4_1}}\label{eq:tdj2_7}
  \end{align}
  Therefore from~\eqref{eq:tdj2_5} and~\eqref{eq:tdj2_7} 
  \begin{align}
    b_{\delta(t)} \in \delta_1\\
    t + b_{\delta(t)} \in t + \delta_1\label{eq:tdj2_8}
  \end{align} 
  Finally, from~\eqref{eq:tdj2_4_1} still,
  \begin{align}
    t  + b_{\delta(t)} &= t + b_{\delta_1} + b - t &\\
    &= t + b_{\delta_1} + b_{\tau_2'} - b_{\delta_1} - t& \te{from the definition of} b\\
    &=  b_{\tau_2'}
  \end{align}
  So $t  + b_{\delta(t)}$ is the smallest value in $\tau_2'$.\\
  Together with \eqref{eq:tdj2_8}, this concludes the proof for this case.
\end{itemize}

\end{proof}

The following result states that the representation $\evalcitd{q}$ is correct:
\begin{restatable}{proposition}{correctTd}\label{prop:correct_td}
  Let $G = \tup{\tdg, \fg, \valz}$ be a temporal graph and $q$ a TRPQ. Then the unfolding of $\evalcitd{q}$ is $\eval{q}$.
\end{restatable}

\begin{proof}\ \\
  Let $G = \tup{\tdg \fg, \valz}$ be a temporal graph,
  and let $q$ be a TRPQ.\\
  We show below that:
  \begin{enumerate}[(I)]
  \item for any $\tup{n_1,n_2,t, d} \in \eval{q}$,
    there are $\tau, \delta \in \intervals{\td}$ such that\label{itemtdcorr_1}
    \begin{enumerate}[(a)]
    \item $\tup{n_1,n_2, \tau, \delta} \in \evalcitd{q}$, \label{itemtdcorr_1_1} 
    \item $t \in \tau$, and \label{itemtdcorr_1_2}
    \item $d \in \delta$.\label{itemtdcorr_1_3}
    \end{enumerate}
  \item for any $\tup{n_1,n_2,\tau, \delta} \in \evalcitd{q}$
    for any $(t, d) \in \tau \times \delta$,\\
    \qquad $\tup{n_1,n_2, t, d}$ is in $\eval{q}$.
    \label{itemtdcorr_2}
  \end{enumerate}
  
    \nt We proceed once again by induction on the structure of $q$.\\
    If $q$ is of the form $\edgelabel$, $\edge^-$, $\pred$, $(\query + \query)$, $\query[m,n]$ or $\query[m,\_]$,
    then~\ref{itemtdcorr_1} and~\ref{itemtdcorr_2} immediately follow from the definitions of
    $\eval{q}$ and 
    $\evalcitd{q}$.\\
    If $q$ is of the form $\neg \node$ or $(?\query)$, then the proof is nearly identical to the one already provided for $\evalcit{q}$.\\
    So we focus below on the two remaining cases:
  \begin{itemize}

\item $q = \query_1 / \query_2$.\ \\
    From the above definitions, we have:
\[
  \begin{array}{rl}
    \eval{q}\ = &\Big\{\tup{n_1, n_3 , t, d_1 + d_2 } \mid \exists n_2 \colon \\
                & \tup{n_1, n_2 , t, d_1} \in  \eval{\query_1} \te{and} \tup{n_2, n_3, t + d_1, d_2} \in  \eval{\query_2}\Big\}\\
\evalcitd{\query_1/\query_2} =& \bigcup \{\u_1 \tjoin \u_2 \mid \u_1 \in \evalcitd{\query_1}, \u_2 \in \evalcitd{\query_2}\}\\
\end{array}
\]
\begin{itemize}
    \item For~\ref{itemtdcorr_1}, let $\vt = \tup{n_1,n_3, t, d} \in \eval{q}$.\\
      Fom the definition of $\eval{q}$, there are  $n_2, d_1$ and $d_2$ such that
      $\tup{n_1, n_2 , t, d_1} \in  \eval{\query_1}$,
      $\tup{n_2, n_3 , t + d_1, d_2} \in  \eval{\query_2}$
      and $d = d_1 + d_2$.\\
      By IH, because $\tup{n_1, n_2, t , d_1} \in  \eval{\query_1}$,
      there are $\tau_1$ and $\delta_1$ such that $t \in \tau_1$, $d_1 \in \delta_1$ and
      \begin{equation}
      \tup{n_1, n_2, \tau_1, \delta_1} \in \evalcitd{\query_1}
    \end{equation}
    Let $R_1$ be the temporal relation induced by this tuple $\tup{n_1, n_2, \tau_1, \delta_1}$.\\
    Since $t \in \tau_1$ and $d_1 \in \delta_1$, we have
    \begin{equation}\label{eq:td2_2_1}
      (t, t + d_1) \in R_1
    \end{equation}

      Similarly,
      because $\tup{n_2, n_3 , t + d_1, d_2} \in  \eval{\query_2}$,
      there are $\tau_2$ and $\delta_2$ such that $t + d_1 \in \tau_2$, $d_2 \in \delta_2$ and
      \begin{equation}
        \tup{n_2, n_3, \tau_2, \delta_2} \in \evalcitd{\query_2}
      \end{equation}
      Let $R_2$ be the temporal relation induced by this tuple $\tup{n_2, n_3, \tau_2, \delta_2}$.\\
    Since $t + d_1 \in \tau_2$ and $d_2 \in \delta_2$, we have
    \begin{equation}\label{eq:td2_2_2}
      (t + d_1, t + d_1 + d_2) \in R_2
    \end{equation}
    So from~\eqref{eq:td2_2_1},~\eqref{eq:td2_2_2} and Lemma~\ref{lemma:tjoin_td},
    there are $\tau$ and $\delta$ such that $\tup{n_1, n_3, \tau, \delta} \in u_1 \tjoin u_2$,
    $t \in \tau$ and $d_1+d_2 = d \in \delta$,
    which concludes the proof for~\ref{itemtdcorr_1}.

  \item For~\ref{itemtdcorr_2}, let $\u = \tup{n_1,n_3,t_1,\delta} \in \evalcitd{q}$, and let $(t,d) \in \tau \times \delta$.\\
    Because $\u \in \evalcitd{q}$, from the definition of $\evalcitd{q}$, there are $\u_1$ and $\u_2$ s.t.:
    \begin{enumerate}[(i)]
    \item $\u \in \u_1 \tjoin \u_2$\label{itemtd3_0_0}
    \item $\u_1 \in  \evalcitd{\query_1}$\label{itemtd3_1}
    \item $\u_2 \in  \evalcitd{\query_2}$\label{itemtd3_2}
    \end{enumerate}
    Let $R_i$ be the temporal relation induced by $u_i$ for $i \in \{1,2\}$.\\
    From~\ref{itemtd3_0_0}, and Lemma~\ref{lemma:tjoin_td},
    \begin{equation}\label{eq:td3_3}
(t, t+d) \in R_1 \join R_2
    \end{equation}
    Now let $\u_1 = \tup{n_1, n_2, \tau_1, \delta_1}$
    and $\u_2 = \tup{n_2, n_3, \tau_2, \delta_2}$ for some $n_2, \tau_1, \tau_2, \delta_1$ and $\delta_2$.\\
    From~\eqref{eq:td3_3} and the definition of $\join$,
    there must be $d_1$ and $d_2$ s.t.~
    $d = d_1 + d_2$, 
$t \in \tau_1, d_1 \in \delta_1$, 
$t + d_1 \in \tau_2$ and $d_2 \in \delta_2$.

So from~\ref{itemtd3_1}, and~\ref{itemtd3_2}, by IH
\begin{align}
  \tup{n_1, n_2, t, d_1} &\in \eval{\query_1}\label{eq:td3_4_1}\\
  \tup{n_2, n_3, t + d_1, d_2} &\in \eval{\query_2}\label{eq:td3_4_2}
\end{align}
So from~\eqref{eq:td3_4_1},~\eqref{eq:td3_4_2} and the definition of $\eval{q}$
\[ \tup{n_1, n_3, t,  d_1 +  d_2} \in \eval{q}, \]
which concludes the proof for~\ref{itemtdcorr_2}.
\end{itemize}

\end{itemize} 
\end{proof}

\subsubsection{In $\tuplestdbe$}

\paragraph{Definition.}
\label{sec:evalbe_def}
If $q$ is an expression for the symbol $\edge$ or $\node$ in the grammar of Section~\ref{sec:language},
then the definition of $\evalcitdbe{q}$ is nearly identical to the one of $\evalcitd{q}$,
extending each tuple $\{\tup{n, n, \tau, [0,0]}$ with $b_\tau$ and $e_\tau$, i.e.
\begin{align*}
  \evalcitdbe{\edge} = & \{\tup{n_1, n_2, \tau, [0,0], b_\tau, e_\tau} \mid  \{\tup{n_1, n_2, \tau, [0,0]} \in \evalcitd{\edge}\}\\
    \evalcitdbe{\node} = & \{\tup{n, n, \tau, [0,0], b_\tau, e_\tau} \mid  \{\tup{n, n, \tau, [0,0]} \in \evalcitd{\node}\}
\end{align*}

Next, if $q$ is of the form
($\query_1 + \query_2$), ($\query[m,\_]$) or ($\query[m,n]$),
then the definition of $\evalcitd{q}$ is once again nearly identical to the one of $\eval{q}$:
\[
  \begin{array}{rl}
  \evalcitdbe{\query_1 + \query_2} = & \evalcitdbe{\query_1} \cup \evalcitdbe{\query_2}\\
  \eval{\query[m,n]} = & \bigcup\limits_{k = m}^n\evalcitdbe{\query^k}\\
  \eval{\query[m,\_]} = &  \bigcup\limits_{k \ge m}\evalcitdbe{\query^k}\\
\end{array}
\]
So the only remaining operator are temporal join ($\query_1/\query_2$) and temporal navigation ($\tndelta$), already defined in the article.
We reproduce these two definition for convenience: 
\[
  \begin{array}{rl}
    \evalcitdbe{\query_1 / \query_2} = & \{\u_1 \tjoin \u_2 \mid \u_1 \in \evalcitdbe{\query_1}, \u_2 \in \evalcitdbe{\query_2} \te{and} \u_1 ~\sim \u_2\}\\
\evalcitdbe{\tndelta} = & \{\tup{n,n,\tdg, \delta,b_{\tdg},e_{\tdg}} \tjoin \tup{n,n,\tdg, [0,0],b_{\tdg},e_{\tdg}} \mid n \in \nodesg \}
\end{array}
\]
where $\u_1 \sim \u_2$ and $\u_1 \tjoin \u_2$ are defined as follows.\\
Let $\u_1 = \tup{n_1, n_2, \tau_1, \delta_1, b_1, e_1}$ and $\u_2 = \tup{n_3, n_4, \tau_2, \delta_2, e_2, b_2}$.\\
Define
\[\delta'_1 =\ \ld{\delta_1}\ b_{\delta_1} + \max(0, b_1 - b_{\tau_1}), e_{\delta_1} - \max(0, e_{\tau_1} - e_1)\ \rd{\delta_1}\]
and
\[\tau =  (((\tau_1 \oplus \delta'_1) \cap \tau_2) \ominus \delta'_1) \cap \tau_1\]
Then $\u_1 \sim \u_2$ iff $n_2 = n_3$ and $\tau \neq \emptyset$.\\
If $\u_1 \sim \u_2$, then $\u_1 \tjoin \u_2 = \tup{n_1, n_4, \tau, \delta_1 \oplus \delta_2, b,e}$,
with
\begin{align*}
  b &= \max(b_1, b_2 - b_{\delta_1})\\
  e &= \min(e_1, e_2 - e_{\delta_1})
\end{align*}

\paragraph{Correctness.}
\label{sec:evalbe_correct}
We start with two lemmas:
\begin{lemma}\label{lemma:tdbe_monot}
  Let $\u = \tup{n_1, n_2, \tau, \delta, b, e}  \in \tuplestdbe$.
Then for any $t_1, t_2 \in \tau$ s.t. $t_1 \le t_2$:
\begin{align*}
t_1 + b_{\delta(t_1)} & \le t_2 + b_{\delta(t_2)}\te{and}\\
t_1 + e_{\delta(t_1)} & \le t_2 + e_{\delta(t_2)}
\end{align*}
\end{lemma}

\begin{lemma}\label{lemma:tdbe_continuous}
Let $\u = \tup{n_1, n_2, \tau, \delta, b, e}  \in \tuplestdbe$.
And let $\tau'$ denote the interval $(b_\tau + b_{\delta(b_\tau)}, e_\tau + e_{\delta(e_\tau} )$.
Then for any $t' \in \tau'$, $\te{there is a} t \in \tau \te{s.t.} t' \in t + \delta(t)$.
\end{lemma}

Next, similarly to what we did above for $\tuplestd$,
if $\u = \tup{n_1, n_2, \tau, \delta, b, e} \in \tuplestdbe$, we call \emph{temporal relation induced by} $\u$ the
set
\[\{(t, t+d) \mid t \in \tau, d \in \delta(t)\}\]
We can now formulate a result analogous to Lemma~\ref{lemma:tjoin_td}:
\begin{lemma}\label{lemma:tjoin_tdbe}
  Let $\u_1, \u_2 \in \tuplestdbe$, and for $i \in \{1,2\}$, let $R_i$ denote the temporal relation induced by $\u_i$.
  If $\u_1 \sim \u_2$ and $\u_1 \tjoin \u_2 = \tup{n_1, n_3, \tau, \delta, b, e}$, then
  \[R_1 \join R_2 =  \{(t, t+d) \mid t \in \tau, d \in \delta(t)\}\]
\end{lemma}
\begin{proof}
  
Let $\u_1 = \tup{n_1, n_2, \tau_1, \delta_1}$ and $\u_2 = \tup{n_2, n_3, \tau_2, \delta_2}$.\\
As explained in Section~\ref{sec:notation}, for $i \in \{1,2\}$ and $t \in \tau_i$,
we use $\delta_i(t)$ for the interval
\[\ld{\delta_{i}}\ b_{\delta_i} + \max(0, b_i - t) \ ,\ e_{\delta_i} - \max(0, t - e_i)\ \rd{\delta_{i}}\]
We need to prove that 
\ei $\tau = \dom{R_1 \join R_2}$ and
that \eii for each $t \in \tau$,
\[t + \delta(t) = \{t' \mid (t,t') \in R_1 \join R_2\}\]
The proof of \ei is nearly identical to the one provided above for Lemma~\ref{lemma:tjoin_td}.\\
\nt For \eii, let $t \in \tau$.\\
\nt We only provide a proof for the case where $\tau$, $\delta_1$ and $\delta_2$ are closed-closed intervals (the proof for the other 63 cases is symmetric).\\
Since $t \in \tau$, from the defnition of $\tau$,
$t \in \tau_1$.\\
Therefore from the definition of $R_1$,
\begin{equation}
  t + \delta_1(t) = \{t' \mid (t,t') \in R_1\}\label{eq:tdbetj_0_0_0}
\end{equation}
So from \ei and the fact that $t \in \tau$
\begin{equation}
  t + \delta_1(t) \cap \dom{R_2} \neq \emptyset
\end{equation}
Now let $a$ (resp. $z$) denote the smallest (resp. largest) value in $t + \delta_1(t) \cap \dom{R_2}$.\\
Then from~\eqref{eq:tdbetj_0_0_0}, $a$ (resp. $z$) is also the smallest value s.t. $(t,a) \in R_1$ and $a \in \dom{R_2}$
(resp.~the largest value s.t. $(t,z) \in R_1$ and $z \in  \dom{R_2}$).\\

\nt Next, from Lemma~\ref{lemma:tdbe_monot}, for any $x \in [a,z]$,
we have 
\begin{equation}
  a + b_{\delta_2(a)} \le x + b_{\delta_2(x)}\label{eq:tdbetj_0_0_1}
\end{equation}
and
\begin{equation}
  x + e_{\delta_2(x)} \le z + e_{\delta_2(z)}\label{eq:tdbetj_0_0_2}
\end{equation}

\nt Now let $a'$ and $z'$ denote $a + b_{\delta_2(a)}$ and $z + e_{\delta_2(z)}$ respectively.\\
From~\eqref{eq:tdbetj_0_0_1} and the definition of $R_2$,
$a'$ is the smallest value s.t. $(x,a') \in R_2$ for some $x \in [a,b]$.\\
And similarly,
from~\eqref{eq:tdbetj_0_0_2} and the definition of $R_2$,
$z'$ is the largest value s.t. $(x,z') \in R_2$ for some $x \in [a,b]$.\\
Together with the defnition of $a$ (resp. of $z$),
this implies that
$a'$ (resp. $z'$) is also the smallest (resp.~largest) value s.t. $(t, a') \in R_1 \join R_2$ (resp.~$(t, z') \in R_1 \join R_2$).\\
To conclude the proof, we show that:
\begin{enumerate}
\item $(t, x) \in R_1 \join R_2$ for each $x \in [a',z']$, and\label{item:tjtdbe_1}
\item $t + \delta(t) = [a',z']$\label{item:tjtdbe_2}.
\end{enumerate}

\nt We start with~\ref{item:tjtdbe_1}.\\
Consider the tuple $\u' = \tup{n_2, n_3, [a,b], \delta_2, b_2, e_2} \in \tuplestdbe$, and let $R'$ be the temporal relation induced by $\u'$.\\
Then from he definitions of $u'$ and $\u_2$: 
\begin{equation}
R' \subseteq R_2\label{eq:tjtdbe_2_1}
\end{equation}
Now take any $x \in [a',z']$.\\
From Lemma~\ref{lemma:tdbe_continuous} and the definitions of $a'$ and $z'$,
there is a $w \in [a,b]$ such that $x \in \delta_2(w)$.\\
Therefore 
\begin{equation*}
(w,x) \in R'
\end{equation*}
So from~\eqref{eq:tjtdbe_2_1} 
\begin{equation}
(w,x) \in R_2\label{eq:tjtdbe_2_2}
\end{equation}
Finally, since $[a,b] = t + \delta_1(t)$ and $w \in [a,b]$,
\begin{equation}
(t,w) \in R_1
\end{equation}
Together with~\eqref{eq:tjtdbe_2_2}, this implies
\begin{equation*}
(t, x) \in R_1 \join R_2
\end{equation*}
which concludes the proof for~\ref{item:tjtdbe_1}.\\

\nt For~\ref{item:tjtdbe_2}, we only prove that $t + b_{\delta_t} = a'$ (the proof that $t + e_{\delta_t} = z'$ is symmetric).\\
Following the definition of $b$,
we consider 2 cases:
\begin{enumerate}
\item $b_1 <  b_2 - b_{\delta_1}$\label{case:1}
\item $b_1 \ge  b_2 - b_{\delta_1}$\label{case:2}
\end{enumerate}
In Case~\ref{case:1}, we have
\begin{align}
&&b_1&<\ \ b_2 - b_{\delta_1}\\
&&\max(b_1, b_2 - b_{\delta_1})&= \ \ b_2 - b_{\delta_1}&\\
&&b&= \ \ b_2 - b_{\delta_1}& \te{from the definition of} b \label{eq:0}
\end{align}
And (in Case~\ref{case:1} still):
\begin{align}
&&b_1&<\ \ b_2 - b_{\delta_1}\\
&&0&<\ \ b_2 - b_{\delta_1} - b_1\\
&&\max(0, b_2 - b_{\delta_1} - b_1)&= \ \ b_2 - b_{\delta_1} - b_1&\label{eq:0_0}
\end{align}
Then we consider two subcases:
\begin{enumerate}[(i)]
\item $t < b_2 - b_{\delta_1}$\label{case:1_1}
\item $t \ge b_2 - b_{\delta_1}$\label{case:1_2}
\end{enumerate}
In Case \ref{case:1_1}, we get\\
  \begin{align}
    && t &<\ \ b_2 - b_{\delta_1} &\\
    && 0 &<\ \ b_2 - b_{\delta_1} - t &\\
    && \max(0,b_2 - b_{\delta_1} - t) &= \ \ b_2 - b_{\delta_1} - t&\label{eq:0_1}
  \end{align}
 Now from the definition of $\delta_t$,
  \begin{align}
&&b_{\delta_t} &=\ \ b_{\delta_1} + b_{\delta_2} + \max(0, b - t)\\
&&&=\ \ b_{\delta_1} + b_{\delta_2} + \max(0, b_2 - b_{\delta_1} - t)& \te{from~\eqref{eq:0}}\\
   &&&=\ \ b_{\delta_1} + b_{\delta_2} + b_2 - b_{\delta_1} - t & \te{from~\eqref{eq:0_1} } \\
   &&&= \ \ b_{\delta_2} + b_2- t\\
     && b_{\delta_t} + t&= \ \ b_{\delta_2} + b_2 - t + t\\
    &&&= \ \ b_{\delta_2} + b_2 \label{eq:1_1}
  \end{align} 
  Next, from the definition of $a'$\\
  \begin{align}
    && a' &=\ \  b_{\delta_2(a)} + a&\\
    && &=\ \  b_{\delta_2} + \max(0, b_2 - a) + a \label{eq:1_2}
  \end{align}
    And, from the definition of $a$
  \begin{align}
    && a &=\ \ b_{\delta_1(t)} + t&\\
    &&&  =\ \ b_{\delta_1} + \max(0, b_1 - t) + t\label{eq:1_3}
  \end{align}
  Then we have two further subcases:
  \begin{enumerate}[(I)]
  \item $t \ge b_1$,\label{eq:case1_1_1} or
  \item $t < b_1$\label{eq:case1_1_2}
  \end{enumerate}
  In case~\ref{eq:case1_1_1}:
  \begin{align}
    && t &\ge\ \ b_1 &\\
    && 0 &\ge\ \ b_1 - t &\\
    && \max(0, b_1 - t) &= \ \ 0&\\
    && a &=\ \ b_{\delta_1} + t & \te{from \eqref{eq:1_3}}\label{eq:1_3_1}\\
    && \max(0, b_2 - a) &=\ \ \max(0, b_2 - b_{\delta_1} - t) &\\
    &&  &=\ \ b_2 - b_{\delta_1} - t & \te{from~\eqref{eq:0_1}}\\
    &&  &=\ \ b_2 - a & \te{from~\eqref{eq:1_3_1}}
  \end{align}
  In case~\ref{eq:case1_1_2}:
  \begin{align}
    && t &<\ \ b_1 &\\
    && 0 &<\ \ b_1 - t &\\
    && \max(0, b_1 - t) &= \ \ b_1 - t&\\
    && a &=\ \ b_{\delta_1} + b_1 - t + t &\te{from~\eqref{eq:1_3}}\\
    && &=\ \ b_{\delta_1} + b_1 \label{eq:1_5}\\
    && \max(0, b_2 - a) &=\ \ \max(0, b_2 - b_{\delta_1} - b_1) &\\
    && &=\ \ b_2 - b_{\delta_1} - b_1 & \te{from~\eqref{eq:0_0}}\\
    && &=\ \ b_2 - a& \te{from~\eqref{eq:1_5}}\\
  \end{align}
  So in both cases~\ref{eq:case1_1_1} and ~\ref{eq:case1_1_2}, we get 
  \[\max(0, b_2 - a) = b_2 - a\]
  Thefore from~\eqref{eq:1_2}
  \begin{align}
    && a' &=\ \ b_{\delta_2} + b_2 - a + a \\
    &&  &=\ \ b_{\delta_2} + b_2 &\\
    &&  &=\ t + b_{\delta_t} &\te{from~\eqref{eq:1_1} } 
  \end{align}
  which concludes the proof for Case~\ref{case:1}~\ref{case:1_1}.\\

  \nt We continue with Case~\ref{case:1}~\ref{case:1_2}.\\
  From~\ref{case:1_2}:
  \begin{align}
    && t &\ge\ \ b_2 - b_{\delta_1} &\\
    && 0 &\ge\ \ b_2 - b_{\delta_1} - t &\\
    && \max(0,b_2 - b_{\delta_1} - t) &= \ \ 0&\label{eq:1_5_1}
  \end{align}
Now from the definition of $\delta_t$:
  \begin{align}
&&b_{\delta_t} &=\ \ b_{\delta_1} + b_{\delta_2} + \max(0, b - t)\\
&&&=\ \ b_{\delta_1} + b_{\delta_2} + \max(0, b_2 - b_1 - t)&\te{from~\eqref{eq:0}}\\
&&&=\ \ b_{\delta_1} + b_{\delta_2} &\te{from~\eqref{eq:1_5_1}}\\
&&b_{\delta_t} + t&=\ \ b_{\delta_1} + b_{\delta_2} + t \label{eq:1_5_2}
  \end{align} 
 Next, from~\ref{case:1} and~\ref{case:1_2}, by transitivity, we get
  \begin{align}
    && b_1 & \le \ \ t&\\
    && \max(0,b_1 - t) & =  \ \ 0\label{eq:1_6}
  \end{align} 
  And from the definition of $a$
\begin{align}
    && a &=\ \ b_{\delta_1(t)} + t\\
    &&  &=\ \ b_{\delta_1} + \max(0,b_1 - t) + t&\\
    &&  &=\ \ b_{\delta_1} + t&\te{from~\eqref{eq:1_6}}\label{eq:1_7}\\
    &&  & \ge \ \ b_{\delta_1} + b_2 -  b_{\delta_1} & \te{from Case~\ref{case:1_2} }\\
    &&  &\ge\ \ b_2\\\
    &&  0&\ge\ \ b_2 - a\\\
    && \max(0, b_2 - a) &=\ \ 0 \label{eq:1_7_0}&
  \end{align}
  Therefore from~\eqref{eq:1_2} and~\eqref{eq:1_7_0}
  \begin{align}
    && a' &=\ \ b_{\delta_2} + a \\
    && &=\ \ b_{\delta_2} + b_{\delta_1} + t&\te{from~\eqref{eq:1_7}}\\
    &&  &=\ b_{\delta_t} + t &\te{from~\eqref{eq:1_1} } 
  \end{align}
  which concludes the proof for Case~\ref{case:1}~\ref{case:1_2}.\\

  \nt We continute with Case~\ref{case:2}.\\
  In this case, we get
\begin{align}
&&b_1&\ge\ \ b_2 - b_{\delta_1}\\
&&\max(b_1, b_2 - b_{\delta_1})&= \ \ b_1&\\
&&b&= \ \ b_1& \te{from the definition of} b \label{eq:1_7_1}
\end{align}
And from~\ref{case:2} still, we derive
\begin{align}
&&b_1&\ge\ \ b_2 - b_{\delta_1}\\
&&0&\ge\ \ b_2 - b_{\delta_1} - b_1\\
&&\max(0, - b_{\delta_1} - b_1)&= \ \ 0&\label{eq:1_7_1_1}
\end{align}
As well as
\begin{align}
&&b_1&\ge\ \ b_2 - b_{\delta_1}&\\
&&b_1 + b_{\delta_1} &\ge\ \ b_2&\label{eq:1_7_1_2}
\end{align}
Next, we distinguish two subcases, namely
\begin{enumerate}[(a)]
\item $t < b_1$\label{case:2_1} and
\item $t \ge b_1$\label{case:2_2}
\end{enumerate}
We start with Case~\ref{case:2_1}.\\
In this case,
  \begin{align}
    && t &< \ \ b_1&\\
    && 0 &<\ \ b_1 - t  &\\
    && \max(0,b_1 - t) &= \ \ b_1 - t&\label{eq:1_7_2}
  \end{align}
And from the definition of $\delta_t$:
  \begin{align}
&&b_{\delta_t} &=\ \ b_{\delta_1} + b_{\delta_2} + \max(0, b - t)\\
&&&=\ \ b_{\delta_1} + b_{\delta_2} + \max(0, b_1 - t)&\te{from~\eqref{eq:1_7_1}}\\
&&&=\ \ b_{\delta_1} + b_{\delta_2} + b_1 - t&\te{from~\eqref{eq:1_7_2}}\\
&&b_{\delta_t} + t&=\ \ b_{\delta_1} + b_{\delta_2} + b_1 - t + t\\
&&&=\ \ b_{\delta_1} + b_{\delta_2} + b_1\label{eq:1_7_2_1}
  \end{align}
  Next, from the definition of $a$
  \begin{align}
    && a & = \ \ b_{\delta_1(t)} + t\\
    && & = \ \ \max(0, b_1 - t) + b_{\delta_1} + t \\
    && & = \ \ b_1 - t + b_{\delta_1} + t & \te{from~\eqref{eq:1_7_2} }\\
    && & = \ \ b_1 + b_{\delta_1} & \label{eq:1_7_3}
  \end{align}
So from~\eqref{eq:1_7_1_2}
  \begin{align}
    && a& \ge\ \  b_2&\\
    && 0& \ge\ \  b_2 - a&\\
    && \max(0, b_2 - a) & = \ \ 0&\\
    && b_{\delta_2} + \max(0, b_2 - a) & = \ \ b_{\delta_2}&\\
    && b_{\delta_2(a)} & = \ \ b_{\delta_2}&\\
    && b_{\delta_2(a)} + a & = \ \ b_{\delta_2} + a &\\
    && a' & = \ \ b_{\delta_2} + a &\te{from the definition of} a'\\
    && a' & = \ \ b_{\delta_2} + b_1 + b_{\delta_1} &\te{from~\eqref{eq:1_7_3}}\\
    && a' & = \ \ b_{\delta_t} + t &\te{from~\eqref{eq:1_7_2_1}}
  \end{align}
  which concludes the proof for Case~\ref{case:2}~\ref{case:2_1}.\\

  \nt We end with Case~\ref{case:2}~\ref{case:2_2}. 
In this case,
  \begin{align}
    && t &\ge \ \ b_1&\\
    && 0 &\ge\ \ b_1 - t  &\\
    && \max(0,b_1 - t) &= \ \ 0&\label{eq:1_8}
  \end{align}
And from the definition of $\delta_t$:
  \begin{align}
&&b_{\delta_t} &=\ \ b_{\delta_1} + b_{\delta_2} + \max(0, b - t)\\
&&&=\ \ b_{\delta_1} + b_{\delta_2} + \max(0, b_1 - t)&\te{from~\eqref{eq:1_7_1}}\\
&&&=\ \ b_{\delta_1} + b_{\delta_2}&\te{from~\eqref{eq:1_8}}\\
&&b_{\delta_t} + t&=\ \ b_{\delta_1} + b_{\delta_2} + t&\label{eq:1_8_1}
  \end{align}
  Next, from the definition of $a$
  \begin{align}
    && a & = \ \ b_{\delta_1(t)} + t\\
    && & = \ \ \max(0, b_1 - t) + b_{\delta_1} + t \\
    && & = \ \ b_{\delta_1} + t & \te{from~\eqref{eq:1_8} }\label{eq:1_8_2}
  \end{align}
    Now from~\ref{case:2_2}
  \begin{align}
    && b_1 + & \le \ \ t& \\
    && b_1 +  b_{\delta_1}& \le \ \ t +  b_{\delta_1}& \\
    && b_1 +  b_{\delta_1}& \le \ \ a& \te{from~\eqref{eq:1_8_2}}\\
    && b_2 & \le \ \ a& \te{from~\eqref{eq:1_7_1_2}, by transitivity}\\
    && b_2 - a  & \le \ \ 0& \\
    &&\max(0, b_2 - a)  & = \ \ 0& \\
    &&b_{\delta_2} + \max(0, b_2 - a)  & = \ \ b_{\delta_2}& \\
    &&b_{\delta_2(a)}  & = \ \ b_{\delta_2}& \\
    &&b_{\delta_2(a)} + a & = \ \ b_{\delta_2} + a& \\
    &&a' & = \ \ b_{\delta_2} + a&\te{from the definition of} a' \\
    &&& = \ \ b_{\delta_2} + b_{\delta_1} + t & \te{from~\eqref{eq:1_8_2}}\\
    &&& = \ \ b_{\delta_t} + t & \te{from~\eqref{eq:1_8_1}}
  \end{align}
  which concludes the proof for Case~\ref{case:2}~\ref{case:2_2}.
\end{proof} 

The following result states that the representation $\evalcitdbe{q}$ is correct:
\begin{restatable}{proposition}{correctTdbe}\label{prop:correct_tdbe}
  Let $G = \tup{\tdg, \fg, \valz}$ be a temporal graph and $q$ a TRPQ. Then the unfolding of $\evalcitdbe{q}$ is $\eval{q}$.
\end{restatable}

\begin{proof}
  
  Let $G = \tup{\tdg \fg, \valz}$ be a temporal graph,
  and let $q$ be a TRPQ.\\
  To prove the result, it is sufficient to show that:
  \begin{enumerate}[(I)]
  \item for any $\tup{n_1,n_2,t, d} \in \eval{q}$,
    there are $\tau, \delta \in \intervals{\td}$ and $b,e \in \td$ such that
    \begin{enumerate}[(a)]
    \item $\tup{n_1,n_2, \tau, \delta, b, e} \in \evalcitdbe{q}$,
    \item $t \in \tau$, and
    \item $d \in \delta(t)$ (where $\delta(t)$ is defined in terms of $t, \delta, b$ and $e$, as explained above).
    \end{enumerate}
  \item for any $\tup{n_1,n_2,\tau, \delta, b, e} \in \evalcitdbe{q}$
    for any $t \in \tau $ and $d \in \delta(t)$,\\
    \qquad $\tup{n_1,n_2, t, d}$ is in $\eval{q}$.
  \end{enumerate}\ \\
    \nt Again, the proof is by induction on the structure of $q$.\\
    If $q$ is of the form $\edgelabel$, $\edge^-$, $\pred$, $\le k$, $(\query + \query)$, $\query[m,n]$ or $\query[m,\_]$,
    then~\ref{itemtdcorr_1} and~\ref{itemtdcorr_2} immediately follow from the definitions of
    $\eval{q}$ and 
    $\evalcitdbe{q}$.\\
    If $q$ is of the form $\neg \node$ or $(?\query)$, then the proof is nearly identical to the one already provided for $\evalcit{q}$.\\
    And if $q$ is of the form $\tndelta$ or $\query_1/\query_2$, then the proof is nearly identical to the one already provided for $\evalcitd{q}$,
    using Lemma~\ref{lemma:tjoin_tdbe} instead of~\ref{lemma:tjoin_td}. 
\end{proof}


\subsection{Complexity of query answering}
\label{sec:complexity}

We provide in this section complexity results for query answering under the different compact representations studied in the article.
The proofs leverage results proven in~\cite{arenas2022temporal} for non-compact answers.

We start by reproducing the decision problem investigated in~\cite{arenas2022temporal}, which will be intrumental:
 \begin{center}
 \fbox{
   \begin{tabular}{l}
     \pbma\\
 \begin{tabular}{ll}
   \textbf{Input}: & \! temporal graph $G$ over discrete time, TRPQ $q$, tuple $\u \in \tuples$\\
   \textbf{Decide}: & $\u \in \eval{q}$
 \end{tabular} \\
 \end{tabular}
}
\end{center}

\vspace{1em}
Our proofs are structured as follows:
\begin{itemize}
\item 
For membership, we leverage the fact that \pbma is in $\PSPACE$, which was proven in~\cite{arenas2022temporal}:
more precisely, we show in Section~\ref{sec:qa_membership} that \pbmt, \pbmd, and \pbmtd can each be reduced to a finite number of independent calls to an oracle for \pbma.
\item  
For hardness, the results trivially immediately from the fact that $\pbma$ is $\PSPACE$-hard,
which was also proven in~\cite{arenas2022temporal}.
We show this in Section~\ref{sec:qa_hardness}, with a basic reduction from $\pbma$ to each of the $4$ other problems.
We also note that the reduction (from QBF to \pbma) provided in~\cite{arenas2022temporal} uses a graph $G$ of fixed size, with the only exception of the temporal domain $\tdg$,
and that the graph $G'$ that we use in our reductions only extends the size of $G$ by a constant factor.
So this property is preserved for our four problems.
\end{itemize}

\subsubsection{Membership}
\label{sec:qa_membership}

If $G = \tup{\tdg \fg, \valz}$ is a temporal graph and $q$ a TRPQ,
we use $\bound{G}{q}$ for the set of all interval boundaries that appear in $G$ and $q$, i.e.
\begin{align*}
  \bound{G}{q} = & \bigcup \Big\{\{b_\delta, e_\delta\} \mid \tndelta \te{appears in} q \Big\}\ \cup  \{b_{\tdg}, e_{\tdg}\}\ \cup \\
                 & \bigcup \Big\{\{b_\tau, e_\tau\} \mid \tau \in \valz{f} \te{for some triple} f \in \fg \Big\}
\end{align*}
Note that $\bound{G}{q}$ is finite.\\

\nt Next, if $Q \subseteq \rat$, we use 
$Q^{+-}$ to denote the least superset of $Q$ that is closed under addition and subtraction.\\
\nt We can now make the two following observations:
\begin{lemma}\label{lemma:closed_boundaries_t}
  Let $G$ be a temporal graph, let $q$ be a TRPQ, let $\u = \tup{n_1, n_2, t, d} \in \eval{q}$, let $Q = \bound{G}{q} \cup \{d\}$, and let
 $\tau$ be the largest interval s.t. $t \in \tau$ and $\tup{n_1, n_2, t', d} \in  \eval{q}$ for all $t' \in \tau$.
  Then \\
   \[b_\tau \in Q^{+-} \te{and} e_\tau \in Q^{+-}\]
 \end{lemma}
\begin{lemma}\label{lemma:closed_boundaries_d}
  Let $G$ be a temporal graph, let $q$ be a TRPQ, let $\u = \tup{n_1, n_2, t, d} \in \eval{q}$, let $Q = \bound{G}{q} \cup \{t\}$ and let
 $\delta$ be the largest interval s.t. $d \in \delta$ and $\tup{n_1, n_2, t, d'} \in  \eval{q}$ for all $d' \in \delta$.
  Then \\
   \[b_\delta \in Q^{+-} \te{and} e_\delta \in Q^{+-}\]
 \end{lemma}
 Next, let $\sqsubseteq$ denote set inclusion lifted to pairs of intervals, i.e.
 \[(\tau_1, \delta_1) \sqsubseteq (\tau_2, \delta_2) \te{iff} \tau_1 \subseteq \tau_2 \te{and} \delta_1 \subseteq \delta_2\]
 The following is an immediate consequence of Lemmas~\ref{lemma:closed_boundaries_t} and~\ref{lemma:closed_boundaries_d}:
 \begin{corollary}\label{cor:closed_boundaries_td}
  Let $G$ be a temporal graph, let $q$ be a TRPQ, let $\u = \tup{n_1, n_2, t, d} \in \eval{q}$, let $Q = \bound{G}{q} \cup \{t, d\}$, let
  $P = \{(t,d) \mid \tup{n_1, n_2, t, d} \in \eval{q}\}$, and let
    $(\tau, \delta) \in \max_\sqsubseteq \{(\tau', \delta') \in \intervals{\td} \times \intervals{\td} \mid t \in \tau \te{and} d \in \delta\}$.
  Then \\
   \[\{b_\tau, e_\tau, b_\delta, e_\delta\} \subseteq Q^{+-}\]
 \end{corollary}
We can now prove our membership results:
\begin{proposition}\label{prop:membership_t}
 \pbmt is in \PSPACE.
\end{proposition}
\begin{proof}\ \\
  Let $G$ be a temporal graph, let $q$ be a TRPQ, and let $\u = \tup{n_1, n_2, \tau, d} \in \tuplest$.\\
  We use $Q$ for $\bound{G}{q} \cup \{d\}$, and $T$  for the set defined by 
    \[T = \{t \mid \tup{n_1, n_2, t, d} \in \eval{q}\}\]
  
 \nt We also use $k$ to denote the product of the denominators of all numbers in $Q$, i.e.
  \[k = \Pi \{j \mid \frac{i}{j} \in Q \te{for some i} \in \integers \}\]
  Note that $k$ (encoded in binary) can be computed in time polynomial (therefore using space polynomial) in the cumulated sizes of $G, q$ and $\u$.\\
  We also use $\ratk$ (resp. $\rattk$) for the set of all multiples of $\frac{1}{k}$ (resp. $\frac{1}{2k}$), i.e
  \[\ratk = \{\frac{i}{k} \mid i \in \integers\}\]
  and
  \[\rattk = \{\frac{i}{2k} \mid i \in \integers\}\]
  Note that
  \[Q^{+-} \subset \ratk \subset \rattk\]

  \nt Now let $\tau'$ be the largest interval such that $\tau \subseteq \tau' \subseteq T$.\\
  Recall that by assumption, $\tau \neq \emptyset$.\\
  Under this assumption,
  $\tup{G,q,u}$ is an instance of \pbmt iff $\tau = \tau'$.\\
  We show that  $\tau = \tau'$ can be decided using space polynomial in the cumulated size of (the encodings of) $G$, $q$ and $u$.\\

  \nt First, 
  fom Lemma~\ref{lemma:closed_boundaries_t},
  we observe that $b_\tau \not\in \ratk$ or $e_\tau \not\in \ratk$ implies $\tau \neq \tau'$.\\
  And $b_\tau \in \ratk$ (resp.~$e_\tau \in \ratk$ ) can be decided in time polynomial in the encoding of $b_\tau$ (resp.~$e_\tau$).\\
  \nt So we can focus on the case where $b_\tau \in \ratk$ and $e_\tau \in \ratk$.\\
  We use $b_\infi$ for the largest element of $(\rattk) \setminus \tau $ that satisfies $b_\infi \le b_\tau$.\\
  And similarly we use $e_\supr$ for the smallest element of $(\rattk) \setminus \tau$ that satisfies $e_\tau \le e_\infi$.\\
  Observe that $b_\infi$ and $e_\supr$ can be computed using space polynomial in (the encoding of) $\tau$.\\ 

  \nt We show below that for any (nonempty) interval $\alpha$ with boundaries in $\ratk$,
  \begin{equation}
    \alpha \subseteq T \te{iff} \alpha \cap \rattk \subseteq T\label{eq:comp_cont}
  \end{equation}
  Therefore in order to decide whether $\tau = \tau'$, it is sufficient to decide whether 
  \begin{enumerate}[(I)]
  \item $\tau \cap \rattk \subseteq T$,\label{item:comp_cont} and
  \item $\{b_\infi, e_\supr\} \cap T = \emptyset$\label{item:comp_bound}
  \end{enumerate}
  Now observe that:
  \begin{itemize}
  \item \ref{item:comp_cont} can be decided with a finite number of independent calls to a procedure for \pbma, and
   \item \ref{item:comp_bound} can be decided with two calls to such a procedure.
   \end{itemize}
   And it was shown in~\cite{arenas2022temporal} that \pbma is in \PSPACE.\\

   \nt To complete the proof, we show that~\eqref{eq:comp_cont} holds.\\
  The right direction ($\alpha \subseteq T$ implies $\alpha \cap \rattk \subseteq T$) is trivial.\\
  For the left direction, assume by contradiction that $\alpha \cap \rattk \subseteq T$ but $\alpha \not\subseteq T$.\\
  Take any $t \in \alpha \setminus T$.\\
  Since $\alpha \cap \rattk \subseteq T$ and $t \not\in T$,
  we have
  \begin{equation}
  t \not\in \rattk\label{eq:comp_4_1}
  \end{equation}

  \nt Next, since $\alpha$ has boundaries in $\ratk$,
  \begin{equation}
  \alpha \cap \rattk \neq \emptyset  
  \end{equation}
  (for instance, $b_\alpha + \frac{1}{2k} \in \alpha \cap \rattk$).\\
  Together with~\eqref{eq:comp_4_1},
  this implies that there is a $t'$ in $\alpha \cap \rattk$ s.t. either $t' < t$ or $t < t'$.\\
  Let us assume w.l.o.g. that the former holds (the proof for the latter case is symmetric).\\
  And let $t_\infi$ be the largest value that satisfies 
  $t_\infi \in \alpha \cap \rattk$ and $t_\infi < t$.\\
  Then
  \begin{equation}
  t - t_\infi < \frac{1}{2k}\label{eq:comp_4_1_1}
  \end{equation}
  Now recall that by assumption, $\alpha \cap \rattk \subseteq T$.\\
  Therefore $t_\infi \in T$.\\
  So from Lemma~\ref{lemma:closed_boundaries_t}, there is a $\beta$ with boundaries in $\ratk$ s.t. $\beta \subseteq T$ and $t_\infi \in \beta$.\\
  Then we have two cases, either $e_\beta \neq t_\infi$ or $e_\beta = t_\infi$:
  \begin{itemize}
  \item ($e_\beta \neq t_\infi$).\\
    In this case,
    since $e_\beta \in \ratk$, and $t_\infi \in \rattk$,
    \begin{equation*}
    \frac{1}{2k} \le e_\beta - t_\infi 
    \end{equation*}
    Together with~\eqref{eq:comp_4_1_1}, this yields (by transitivity) 
\begin{align}
  t - t_\infi &< e_\beta - t_\infi\\
  t  &< e_\beta\label{eq:comp_4_2_1}
\end{align}
Now since $t_\infi \in \beta$,
  \begin{equation}
  b_\beta \le t_\infi
\end{equation}
Together with $t_\infi < t$, this implies
\begin{equation}
b_\beta < t
\end{equation}
Together with~\eqref{eq:comp_4_2_1}, this yields
\begin{equation*}
  t \in \beta
\end{equation*}
Since $\beta \subseteq T$,
this implies $t \in T$, which contradicts the definition of $t$.\\

  \item ($e_\beta = t_\infi$).\\
    In this case,
    since $\beta$ has boundaries in $\ratk$,
    \begin{equation}
    t_\infi \in \ratk\label{eq:comp_4_3}
  \end{equation}
  And because $t \in \alpha$ and $t_\infi < t$
    \begin{align}
    t_\infi &< t \le e_\alpha\\
    t_\infi &< e_\alpha\label{eq:comp_4_3_0}
  \end{align}
  Together with~\eqref{eq:comp_4_3} and $e_\alpha \in \ratk$, this implies
    \begin{equation}
   \frac{1}{k} \le e_\alpha - t_\infi\label{eq:comp_4_3_1}
    \end{equation}
    Now let $t_\supr = t_\infi + \frac{1}{2k}$.\\
    From~\eqref{eq:comp_4_3_1}, we get
    \begin{equation}
      t_\supr < e_\alpha\label{eq:comp_4_3_1_1}
    \end{equation}
    Next, since $t_\infi \in \alpha$ and $t_\infi < t_\supr $
    \begin{equation}
      b_\alpha < t_\supr\label{eq:comp_4_3_1_2}
    \end{equation}
    So from~\eqref{eq:comp_4_3_1_1} and~\eqref{eq:comp_4_3_1_2}
    \begin{equation*}
      t_\supr \in \alpha
    \end{equation*}
    So from Lemma~\ref{lemma:closed_boundaries_t}, there is a $\beta'$ with boundaries in $\ratk$ s.t. $\beta' \subseteq T$ and $t_\supr \in \beta'$.\\
    Next, from ~\eqref{eq:comp_4_1_1} and the definition of $t_\supr$
    \begin{equation}
    t_\supr - t < \frac{1}{2k} 
  \end{equation}
  So with an argument symmetric to the one used above to show $t \in \beta$, we get $t \in \beta'$,
  which once again contradicts $t \not\in T$.
\end{itemize} 
\end{proof}

\begin{proposition}\label{prop:membership_d}
 \pbmd is in \PSPACE
\end{proposition}
\begin{proof}\ \\
  The proof is symmetric to the one provided above for Proposition~\ref{prop:membership_t}, using Lemma~\ref{lemma:closed_boundaries_d} instead of Lemma~\ref{lemma:closed_boundaries_t}.  
\end{proof}

\begin{proposition}\label{prop:membership_td}
 \pbmtd is in \PSPACE.
\end{proposition}
\begin{proof}\ \\
  The proof is analogous to the one provided above for Proposition~\ref{prop:membership_t}, using Corollary~\ref{cor:closed_boundaries_td} instead of Lemma~\ref{lemma:closed_boundaries_t}. \\
  More precisely,
  let $G$ be a temporal graph, let $q$ be a TRPQ, and let $\u = \tup{n_1, n_2, \tau, \delta} \in \tuplestd$.\\
  We use $Q$ for $\bound{G}{q} \cup \{t, d\}$, and $P$  for the set defined by 
    \[P = \{(t,d) \mid \tup{n_1, n_2, t, d} \in \eval{q}\}\]

    \nt We also define  $k$, $\ratk$ and $\rattk$ identically as in the proof of Proposition~\ref{prop:membership_t}.\\
    Then  analogously to what we showed in this proof,
  for any pair of intervals $(\alpha_1, \alpha_2)$ with boundaries in $\ratk$,
    \[\alpha_1 \times \alpha_2 \subseteq P \te{iff} (\alpha_1 \cap \rattk) \times (\alpha_2 \cap \rattk) \subseteq P\]
   So with a similar argument, deciding whether $\tup{n_1, n_2, \tau, \delta}$ is a compact answer to $q$ over $G$
   can be reduced to deciding
  \begin{itemize}
  \item $\{b_\tau, e_\tau, b_\delta, e_\delta\} \subseteq \ratk$,
  \item $(\tau \cap \rattk) \times (\delta \cap \rattk) \subseteq P$,
  \item
    $\left(\{b^\tau_\infi, e^\tau_\supr\} \times (\delta \cap \rattk)\right) \cap P = \emptyset$ and
  \item
    $\left( (\tau \cap \rattk) \times \{b^\delta_\infi, e^\delta_\supr\}  \right) \cap P = \emptyset$
  \end{itemize}
  where
  $b^\tau_\infi$ is the largest element in $(\rattk) \setminus \tau$ that satisfies $b^\tau_\infi \le b_\tau$,
  $e^\tau_\supr$ is the smallest element in $(\rattk) \setminus \tau$ that satisfies $e_\tau \le  e^\tau_\supr$,
  and 
  $b^\delta_\infi$ and $e^\delta_\supr$ are defined analogously.
\end{proof}


  

\subsubsection{Hardness}
\label{sec:qa_hardness}

\begin{proposition}
 \pbmt is \PSPACE-hard
\end{proposition}
\begin{proof}
  The proof is a straightforward reduction from \pbma.\\
  Let $G = \tup{\tdg \fg, \valz}$
  be a temporal graph, let $q$ be a TRPQ and let $\u = \{n_1, n_2, t, d\} \in \tuples$.\\
  W.l.o.g., let us assume that $\{n_1, n_2\} \subseteq \nodesg$ (the proof for the 3 other cases is symmetric). \\
  
  \nt Now
  consider two fresh edge labels $e_1, e_2 \in \E$ that do not appear in $\fg$.
  And let $G'= \tup{\tdg, \fgp, \valpz}$ be the graph nearly identical to $G$,
  but extended with two triples,
  as follows:
  \begin{itemize}
  \item $\valpz(f) = \valz(f)$ for each triple $f \in \fg$,
  \item $\valp{n_1}{e_1}{n_1} = \{[t,t]\}$, and
  \item $\valp{n_2}{e_2}{n_2} = \{[t+d,t+d]\}$.
  \end{itemize}
  Let $q'$ be the following TRPQ:
  \[q' = e_1/q/e_2\]
  Then immediately from the semantics of TRPQs:
  \begin{equation}\label{eq:comp_hard_1}
    \u \in \eval{q} \te{iff} \evalp{q'} = \{\tup{n_1, n_2, t, d}\}
  \end{equation}
  Now consider the tuple $\u' = \{n_1, n_2, [t,t], d\} \in \tuplest$.\\
  Then from~\eqref{eq:comp_hard_1}, $\u \in \eval{q}$ iff $\u'$ is a compact answer to $q$ over $G'$ in $\tuplest$.\\

  \nt Clearly, this input $\tup{G', q', \u'}$ for \pbmt can be computed in time polynomial in the size of (the encodings of) $G, q$ and $\u$.\\ 
  And it was shown in~\cite{arenas2022temporal} that \pbma is \PSPACE-complete.
  
\end{proof}

\begin{proposition}\ \\
 \pbmd, \pbmtd and \pbmtdbe are \PSPACE-hard.
\end{proposition}
\begin{proof}
  The proofs are nearly identical to the one provided above for \pbms.\\
  The graph $G'$ is defined identically in all cases, so that the reductions only differ w.r.t. to the tuple $\u'$.\\
  This tuple is defined as follows:
  \begin{itemize}
    \item $\tup{n_1, n_2, t, [d,d]}$ for \pbmd,
    \item $\tup{n_1, n_2, [t,t], [d,d]}$ for \pbmtd,
  \item $\tup{n_1, n_2, [t,t], [d,d], t,t}$ for \pbmtdbe.
  \end{itemize}
\end{proof}






\subsection{Size of compact answers}
\label{sec:size}

\propSizeEvalSingleDimension*
\begin{proof}\ \\
\begin{itemize}
  \item 

  $\sizea{\tau}{\tuplest} = O(1)$ follows from the definition of $\evalcit{q}$ (in Section~\ref{sec:evalt_def}) and Proposition~\ref{prop:correct_t}, by induction on the structure of $q$:
    \medskip
  \begin{itemize}
  \item If $q$ is of the form  $\edgelabel$, $\edge^-$, $\pred, \le k, \tndelta$ or $?\query$, the number of tuples in $\evalcit{q}$ is independent of the size of the intervals present in $G$.

  \item For the binary operators, (i.e.~when $q$ is of the form $\query_1 / \query_2$ or $\query_1 + \query_2$), the cardinality of $\evalcit{q}$ is bounded by some function of the size of the operands.\\
    Precisely, if $k_1$ (resp.~$k_2$) is the cardinality of $\evalcit{q_1}$ (resp. $\evalcit{q_2}$),
    then the cardinality of $\evalcit{q_1 /q_2}$ (resp. $\evalcit{q_1 + q_2}$) is bounded by $k_1 \cdot k_2$ (resp.~$k_1 + k_2$).
    And by IH, the cardinality of $\evalcit{q_i}$ is independent of the size of intervals in $G$.\\
    This argument also applies to the case where $q$ is of the form $\query[m,n]$, since it is equivalent in this case to a finite union of joins.
  \item  If $q$ is of the form $\neg q'$, then the cardinality of $\evalcit{q}$ is bounded by $2k$, where $k$ is the cardinality of $\evalcit{q'}$,
    since the complement in $\tdg$ of an interval can always be represented with at most two intervals.
    And once again, by IH, $k$ is independent of the size of intervals in $G$. 
  \end{itemize}
  \medskip
\item 
  To show that $\sizea{\delta}{\tuplest} = \Omega(n)$,
  for each $i \in \nn$, consider the TRPQ $q_i = \text{T}_{[0,i]}$,
  and the graph $G_i$ with (discrete) effective temporal domain $[0,i]$,
  a single node $n$ and no edge.
  Then \[\evalcit[G_i]{q_i} = \{\tup{n,n,[0,0],d} \mid d \in [0, i]\}\]
  Each tuple in this set has a different value for the distance $d$, therefore $\evalcit[G_i]{q_i}$ is the compact representation of $\eval[G_i]{q_i}$ in $\tuplest$.\\
  And this set has cardinality $i$.
  \medskip
\item
  $\sizea{\delta}{\tuplesd} = O(1)$ follows from the definition of $\evalcid{q}$ (in Section~\ref{sec:evald_def}) and Proposition~\ref{prop:correct_d},
  with an argument nearly identical to the one that we provided above to show that $\sizea{\tau}{\tuplest} = O(1)$.
  \medskip
\item To show that $\sizea{\tau}{\tuplesd} = \Omega(n)$,
  we use a fixed query $q = e$,
  with $e \in \E$.
  For each $i \in \nn$, consider 
  the graph $G_i$ with (discrete) effective temporal domain $[0,i]$,
  and a single triple  $f = (n_1, e, n_2)$ 
  such that $\valz[G_i](f) = \{[0,i]\}$.
Then
\[\evalcid[G_i]{q} = \{\tup{n_1,n_2,t,[0,0]} \mid t \in [0, i]\}\]
Each tuple in this set has a different value for the distance $d$, therefore $\evalcid[G_i]{q}$ is the compact representation of $\eval[G_i]{q}$ in $\tuplesd$.\\
And this set has cardinality $i$.
\end{itemize}
\end{proof}

\propSizeEvalBothDimensions*
\begin{proof}\ \\
 \begin{itemize}
   \item 
  To show that $\sizea{\delta}{\tuplesd} = \Omega(n)$, we use an example analogous to the one used to show non-finiteness (illustrated with Figure~\ref{fig:rectangles}b).
 
  For each $i \in \nn$, consider the TRPQ $q_i = \text{T}_{[0,i]}$,
  and the graph $G_i$ with (discrete) effective temporal domain $[0,i]$,
  a single node $n$ and no edge.
  Then 
  \[\eval[G_i]{q_i} = \{\tup{n,n,t,d} \mid t \in [0, i] \te{and} d \in [0, i - t]\}\]
  For instance, if $i = 2$, then
  \[\eval[G_i]{q_i} = \{
  \tup{n,n,0,0},
  \tup{n,n,0,1},
  \tup{n,n,0,2},
  \tup{n,n,1,1}, 
  \tup{n,n,1,2}, 
  \tup{n,n,2,2}
  \}\]
  Now consider again the Euclidean plane $\integers \times \integers$ of time per distance,
  and the polygon defined by the points $(0,0), (0,n), (n,0)$ and $\{ (t, i-t + 1) \mid t \in [1, i]\}$.\\
  Intuitively, this is a ``near triangle'' (visually, analogous to an upper approximation of the integral of a linear function).\\
  The minimal number of tuples in $\tuplestd$ needed to represent $\eval[G_i]{q_i}$ is the minimal number of rectangles needed to cover this polygon,
  and this number is $i$.
\medskip
   \item 
  To show that $\sizea{\tau}{\tuplestd} = O(1)$
 we use an argument similar to the one that we provided to show that $\sizea{\tau}{\tuplest} = O(1)$, in the proof of Proposition~\ref{prop:size_eval_single_dimension}.

  We proceed once again again by induction on the structure of $q$.
  If $q$ is an expression for the symbols $\edge$ or $\node$ in the grammar of Section~\ref{sec:language},
  then the argument is identical to the one provided for $\tuplest$.
  This is is also the case if $q$ is of the form $\query + \query$ or $\query[m,n]$.\\
  So we focus here on the two remaining operators, namely the cases where $q$ is of the form $\tndelta$ or $\query_1 / \query_2$.\\
  Recall that
  \[\evalcitd{\query_1/\query_2} \bigcup \{\u_1 \tjoin \u_2 \mid \u_1 \in \evalcitd{\query_1}, \u_2 \in \evalcitd{\query_2}\}\]
  and
  \[\evalcitd{\tndelta} = \bigcup_{n \in \nodesg} \{\tup{n,n,\tdg, \delta} \tjoin \tup{n,n,\tdg, [0,0]}\}\]
  Let 
  $\u_1 = \tup{n_1, n_2, \tau_1, \delta_1}$ and
  $\u_2 = \tup{n_1, n_2, \tau_2, \delta_2}$ be two tuples in $\tuplestd$.
  For simplicity,
  we assume that all intervals here are closed-closed, and we focus on the discrete case (but the argument can be easily adapted to the other cases).\\
   We show that
   the cardinality of $\u_1 \tjoin \u_2$ is bounded by a function of the cardinality of $\delta_1$.
   Together with the definitions of $\evalcitd{\query_1/\query_2}$ and 
   $\evalcitd{\tndelta}$, this proves our claim.\\
  \nt We use
   $R_1, R_2$,
   $R_1 \join R_2$, $b$ and $e$ here with the same meaning as in the proof of Lemma~\ref{lemma:tjoin_td}.
   We already saw that
   the cardinality of $\u_1 \tjoin \u_2$ is the cardinality of
   the set $\dom{R_1 \join R_2} \setminus [b,e] + 1$ if $b \le e$,
   and the cardinality of $\dom{R_1 \join R_2}$ otherwise.\\
   In the former case,
   the set $\dom{R_1 \join R_2} \setminus [b,e]$ is the union of the two intervals
   $[b_{\dom{R_1 \join R_2}}, b]$ and 
   $[e, e_{\dom{R_1 \join R_2}}]$,
   and the cardinality of each of these is bounded by $e_{\delta_1} - b_{\delta_1}$,
   which is indeed the cardinality of $\delta_1$.\\
   For the case where 
   $e < b$,
   $\dom{R_1 \join R_2}$ is also the union of these two intervals (which now overlap),
   and 
   $e_{\delta_1} - b_{\delta_1}$ is still an upper bound on the cardinality of each of them.
\medskip
  \item $\sizea{\tau}{\tuplestd} = O(1)$ and $\sizea{\tau}{\tuplestd} = O(1)$ 
  both 
  follow from the definition of $\evalcitdbe{q}$ (in Section~\ref{sec:evalbe_def} ) and Proposition~\ref{prop:correct_tdbe},
  with arguments nearly identical to the one that we provided to show that $\sizea{\tau}{\tuplest} = O(1)$,
  in the proof of Proposition~\ref{prop:size_eval_single_dimension}.
 \end{itemize}
\end{proof}

\subsection{Compactness}
\label{sec:minimization}
Regarding the cost of coalescing, all our results are already justified in the body of the article.

Regarding (non-)unicity of compact representations, all arguments are also provided,
with the exception of non-unicity for the fourth representations (i.e.~in $\tuplestdbe$).
For this case, we show that a set $V$ of tuples in $\tuples$ that share the same nodes $n_1$ and $n_2$ may have several compact representations in $\tuplestdbe$.

Over dense time, consider the Euclidean plane $\rat \times \rat$.
Observe that:
  \begin{enumerate}[(i)]
\item any rectangle $\tau \times \delta$ in this plane is exactly covered by some tuple in $\tuplestdbe$
  (namely $\tup{n_1,n_2, \tau, \delta, b_\tau, e_\tau}$),~\label{enum:compact1} and
\item the area covered by a tuple in $\tuplestdbe$ forms either a rectangle, or a polygon with some non-square angles.\label{enum:compact2}
\end{enumerate}

Now assume that the area covered by $V$ forms an $L$-shaped polygon.
From~\ref{enum:compact2}, this area cannot be exactly covered by a single tuple in $\tuplestdbe$,
and from~\ref{enum:compact1} there is more than one pair of tuples that exactly cover it (as illustrated with Figure~\ref{fig:rectangles}a).

The same argument can easily be adapted to discrete time.




\end{document}